\documentclass[11pt]{article}
\usepackage{amsmath,amsthm,amssymb}
\usepackage{graphicx}
\usepackage{tikz}
\usepackage{float}
\usetikzlibrary{decorations.pathreplacing}
\usepackage{hyperref}
\hypersetup{colorlinks,citecolor=blue,urlcolor=blue,linkcolor=blue}
\usepackage{natbib}
\setcitestyle{numbers,square}
\usepackage{bbm}
\usepackage{algorithm}
\usepackage{algorithmic}
\usepackage{caption}
\usepackage{subcaption}
\usepackage{hhline}
\usepackage[inline, shortlabels]{enumitem}

\usetikzlibrary{calc}

\usepackage{my_style}

\newcommand{\FDP}{\operatorname{FDP}}
\newcommand{\FDR}{\operatorname{FDR}}
\newcommand{\pval}{p\operatorname{-value}}
\newcommand{\pvals}{p\operatorname{-values}}
\newcommand{\Ho}{\mathcal{H}_0}

\newcommand{\Vopt}{V_{\operatorname{opt}}}
\newcommand{\bX}{\mathbf{X}}
\newcommand{\bXb}{\mathbf{X}^{(b)}}
\newcommand{\bY}{\mathbf{Y}}
\newcommand{\Unif}{\operatorname{Unif}}

\title{Deploying the Conditional Randomization Test\\ in High
  Multiplicity Problems} \author{Shuangning Li \thanks{Department of
    Statistics, Stanford University, USA.} \and Emmanuel J.~Cand\`es
  \thanks{Departments of Statistics and of Mathematics, Stanford
    University, USA.}  }

\begin{document}

\maketitle

\begin{abstract}
  This paper introduces the sequential CRT, which is a variable selection procedure that combines the conditional randomization test (CRT) and Selective SeqStep+. Valid $\pvals$ are constructed via the flexible CRT, which are then
  ordered and passed through the selective SeqStep+ filter to produce
  a list of discoveries. We develop theory guaranteeing control on the
  false discovery rate (FDR) even though the $\pvals$ are not
  independent.  We show in simulations that our novel procedure
  indeed controls the FDR and are competitive with---and sometimes
  outperform---state-of-the-art alternatives in terms of
  power. Finally, we apply our methodology to a breast cancer dataset
  with the goal of identifying biomarkers associated with cancer
  stage.
\end{abstract}

\section{Introduction}
To quote from \citet{benjamini2014discussion},
\begin{quote}
Significance testing is an effort to address the selection of an
interesting finding regarding a single parameter from the background
noise. Modern science faces the problem of selection of promising
findings from the noisy estimates of many.
\end{quote}
This paper is about the latter. In contemporary studies, geneticists
may have measured hundreds of thousands of genetic variants and wish
to know which of these influence a trait \citep{sesia2019gene,
  sesia2020multi}. Scientists may be interested in discovering which
demographic and clinical variables influence the susceptibility to
Parkinson’s disease \citep{gao2018model}. Economists study which
variables from individual employment and wage histories affect future
professional careers \citep{klose2020pipeline}. In all these examples
and countless others, we have hundreds or even thousands of
explanatory variables and are interested in determining which of these
influence a response of interest. The problem is to select
associations which are replicable, that is, without having too many
false positives.

Formally, let $Y \in \mathbb{R}$ be the response we wish to study, and
$X = (X_1, X_2, \ldots, X_p) \in \mathbb{R}^p$ be the vector of
explanatory variables. We call variable $j$ a null variable if $X_j$
is conditionally independent of $Y$ given the other $X$'s. This says
that the $j$-th variable does not provide information about the
response beyond what is already provided by all the other variables
(roughly, if it is not in the Markov blanket of $Y$).  Expressed
differently, a variable is null if and only if the hypothesis
\begin{equation}
\label{eqn:indep}
\mathcal{H}_j: X_j \independent Y | X_{-j}, 
\end{equation}
is true. (Throughout, $X_{-j}$ is a shorthand for all $p$ variables
except the $j$th.) Likewise, a variable $j$ is nonnull if
$\mathcal{H}_j$ is false.  Let
$\mathcal{H}_0 \subset \cb{1, \dots, p}$ be the subset of
nulls. Suppose now we have $n$ independent samples assembled in a data
matrix $\mathbf{X} \in \mathbb{R}^{n \times p}$ and a response vector
$\mathbf{Y} \in \mathbb{R}^n$. The goal is to identify the nonnull
variables with some form of type-I error control. Specifically, we
consider in this paper the false discovery rate (FDR)
\citep{benjamini1995controlling}, namely, the expected fraction of
false positives defined as
\begin{equation}
\FDR = \EE{\frac{|\hat{\mathcal{S}} \cap \mathcal{H}_0|
  }{|\hat{\mathcal{S}}| \vee 1  }},\footnote{Here and below, $a \vee b = \max(a,b)$ and
  $a\wedge b = \min(a,b)$}.
  \end{equation}
where $\hat{\mathcal{S}}$ is the selected set of variables.

\subsection{The conditional randomization test}
Naturally, in order to identify the nonnull variables, one could test
the hypotheses $\mathcal{H}_j$ in \eqref{eqn:indep}.
\citet{candes2018panning} proposed to achieve this via the conditional
randomization test (CRT). To run the CRT, we resample $\bX_j$---the
$j$th column of the matrix $\bX$---conditional on the other variables,
calculate the value of a test statistic, and compare it to the test
statistic computed on the true $\bX_j$. When the statistic computed on
the true $\bX_j$ has a high rank when compared with those obtained
from imputed values, this is evidence against the null. Details of the
CRT are given in Algorithm~\ref{alg:crt}. There, the output p-value is
valid in the sense that under the null, it is stochastically larger
than a uniform variable.
\begin{algorithm}[h]
\caption{Conditional Randomization Test (CRT)}
\label{alg:crt}
\begin{algorithmic}
\REQUIRE Data $(\bX,\bY)$, test statistic $T(\cdot)$, number of randomizations $B$
\FOR{$b \in \cb{1, \dots, B}$}
	\STATE Sample $\bX_j^{(b)}$ from the distribution of $\bX_j | \bX_{-j}$, independently of $\bX_j$ and $\bY$.
\ENDFOR
\ENSURE  The $\pval$
\[ p_j=\frac{1}{B+1}
	\p{1 + \sum_{b = 1}^B \one{\cb{T(\bX_j, \bX_{-j}, \bY) \leq T(\bX_j^{(b)}, \bX_{-j}, \bY)}}}.  \footnotemark
	\] 
\end{algorithmic}
\end{algorithm}
\footnotetext{If some values of the test statistics are the same, we break ties randomly. }

Informally, under the null hypothesis that the variable $X_j$ is
independent of $Y$ conditional on $X_{-j}$, each one of the new
samples $\bXb_j$ has the same distribution as $\bX_j$, and they are
all independent conditionally on $\bY$ and $\bX_{-j}$. As a
consequence, each $T(\bXb_j, \bX_{-j}, \bY)$ has the same distribution
as $T(\bX_j, \bX_{-j}, \bY)$. Thus the rank of
$T(\bX_j, \bX_{-j}, \bY)$ among $\cb{T(\bXb_j, \bX_{-j}, \bY)}$ will
be uniform in $\cb{1, \dots, B}$ assuming we break ties at
random. Formally, we have:
\begin{theo}[\citet{candes2018panning}]
 If $X_j \independent Y | X_{-j}$, then the $\pvals$ from Algorithm
 \ref{alg:crt} satisfy  $\PP{p_j \leq \alpha} \leq \alpha$, for any
 $\alpha \in [0,1]$. This holds regardless of the test statistic
 $T(\cdot)$. 
\end{theo}
The validity of the procedure does not rely on any assumptions on the
distribution of $Y|X$, parametric or not. Yet it requires knowledge of
the distribution of the covariates $X$. This is known as the Model-X
framework, and is an appropriate assumption in many important
applications, including genetic and economics studies, where either
knowledge about the exact covariate distribution or a large amount of
unsupervised data of covariates is available \citep{cong2013multiplex,
  haldane1931inbreeding, peters2016comprehensive, saltelli2008global,
  tang2006reconstructing}. Rapid progress has been made on
methodological advances in this framework \citep{candes2018panning,
  tansey2018holdout, sesia2019gene, bates2020metropolized,
  liu2020fast, romano2020deep, ren2020derandomizing} and in applications
to genetic studies \citep{sesia2019gene, sesia2020multi,
  bates2020causal}.

\subsection{Selective SeqStep+}
We still need a selection procedure that transforms the CRT $\pvals$
into a selected set with FDR control guarantees. A natural choice of
variable selection procedure is the Benjamini-Hochberg procedure (BHq)
\citep{benjamini1995controlling}. If we were to apply BHq, we would
need to compare the $\pvals$ with critical thresholds of the form
$\alpha_i = i q/p$, where $q$ is the nominal FDR level.  When $p$ is
large, $q$ is $0.1$, and $i = 1,2, \dots$, this requires $\pvals$ on
an extremely fine scale.  The $\pvals$ defined in Algorithm
\ref{alg:crt}, however, can only take values in
$\{1/(B+1), 2/(B+2), \dots, 1\}$, thus a huge number of randomizations
in CRT is required.  This makes the combination of CRT and BHq
computationally expensive or even infeasible. This is the motivation
for this paper: can we find a
selection procedure that does not require any of the $\pvals$ to be
very small and works well with discrete $\pvals$?

\begin{algorithm}[t]
\caption{Selective SeqStep+}
\label{alg:seqstep}
\begin{algorithmic}
\REQUIRE A sequence of $\pvals$ $p_1, \dots p_p$
\STATE 
Let
\begin{equation}
\label{eqn:seqstep}
\hat{k} =\max \left\{k \in \cb{1,\dots, p}: \frac{1+\#\left\{j \leq k: p_{j}>c\right\}}{\#\left\{j \leq k: p_{j} \leq c\right\} \vee 1} \leq \frac{1-c}{c} \cdot q\right\}. 
\end{equation}
\ENSURE  Selected set of nonnulls $\hat{\mathcal{S}} = \cb{ j \leq \hat{k}: p_{j} \leq c}$. 
If the set in $\eqref{eqn:seqstep}$ is empty, $\hat{\mathcal{S}}$ is the empty set as well.  
\end{algorithmic}
\end{algorithm}

To this end, we consider SeqStep+, a sequential testing procedure
first introduced by \citet{barber2015controlling}. We consider a
specific version, namely, Selective SeqStep+, which takes
a sequence of $\pvals$ $p_1, \dots p_p$ as input, and outputs a
selected set $\hat{\mathcal{S}}$. The procedure starts by finding an
integer $\hat{k}$ such that among the $\pvals$
$\cb{p_1, \dots, p_{\hat{k}}}$, few are greater than a user-specified
threshold $c$. In details, $\hat{k}$ is the largest $k$ in
$\cb{1, \dots p}$ such that the ratio between
$1+\#\left\{j \leq k: p_{j}>c\right\}$ and
$\#\left\{j \leq k: p_{j} \leq c\right\} \vee 1$ is no greater than
$(1-c)q/c$. The procedure then selects all $j$'s, such that
$j \leq \hat{k}$ and $p_j \leq c$. We include details of the procedure
in Algorithm~\ref{alg:seqstep}.

To understand why the procedure works, assume that the null $\pvals$ are $\text{i.i.d.}\Unif[0,1]$. Then, the ratio of $\#\left\{\text {null } j \leq \hat{k}: p_{j} \leq c\right\}$ to $\#\left\{\text {null } j \leq \hat{k}: p_{j}>c\right\}$ is roughly $c/(1-c)$. Hence 
\begin{align*}
\FDP &= \frac{\#\left\{\text {null } j \leq \hat{k}: p_{j} \leq c\right\}}{\#\left\{j \leq \hat{k}: p_{j} \leq c\right\} \vee 1}  
\approx \frac{c}{1-c} \cdot \frac{\#\left\{\text {null } j \leq \hat{k}: p_{j}>c\right\}}{\#\left\{j \leq \hat{k}: p_{j} \leq c\right\} \vee 1} \\
&  \leq \frac{c}{1-c} \cdot \frac{\#\left\{j \leq \hat{k}: p_{j}>c\right\}}{\#\left\{j \leq \hat{k}: p_{j} \leq c\right\} \vee 1}
\leq \frac{c}{1-c} \cdot \frac{1-c}{c}q = q. 
\end{align*}
Formally, we have the following result. 
\begin{theo}[\citet{barber2015controlling}]
\label{theo:indep_FDR}
Assume that the ordering of the $\pvals$ is fixed. If all null $\pvals$ are independent with $p_j \geq \operatorname{Unif} [0,1]$, and are independent from the nonnulls, then Selective SeqStep+ controls the FDR at level $q$. 
\end{theo}

A close look at \eqref{eqn:seqstep} shows that the only information
Selective SeqStep+ uses from the $\pvals$ is whether or not
$p_j \leq c$. This means that unlike BHq, Selective SeqStep+ does not
require some of the $\pvals$ to be very small to make rejections, and
hence would require a much smaller number of randomizations $B$.
Selective SeqStep+ would, therefore, be computationally far less
intensive.

\subsection{Challenges and our contribution}
\label{subsection:introduction_challenges}

In this paper, we study variable selection procedures with the
conditional randomization test and Selective SeqStep+. There are three challenges and we address them all. 

The first challenge is in the dependency of the $\pvals$. The $\pvals$
from CRT are not independent in general, hence Theorem
\ref{theo:indep_FDR} does not apply. In response, we will develop
theory in Section \ref{section:theory} showing how we can make
SeqStep+ valid under dependence. In particular, we will show examples
of approximate FDR control when the $\pvals$ are weakly dependent or
when they are exchangeable in distribution.

\begin{figure}
\begin{subfigure}{\textwidth}
\caption{A good ordering}
\includegraphics[width = \textwidth]{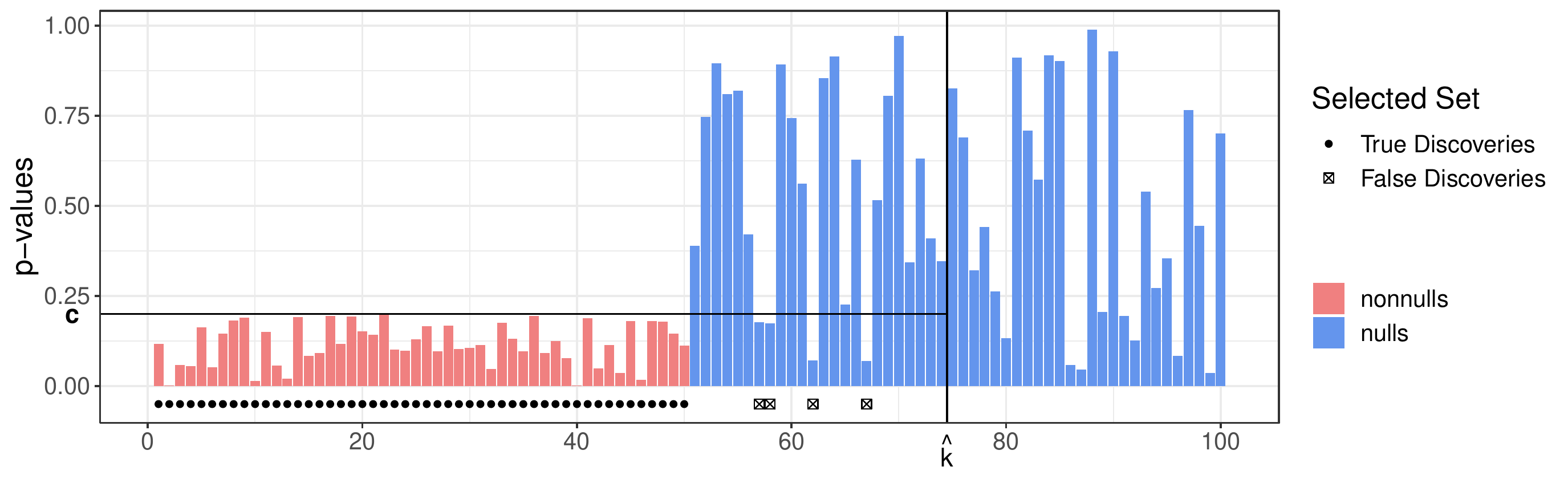}
\label{fig:good_ordering}
\end{subfigure}
\begin{subfigure}{ \textwidth}
\caption{A bad ordering}
\includegraphics[width = \textwidth]{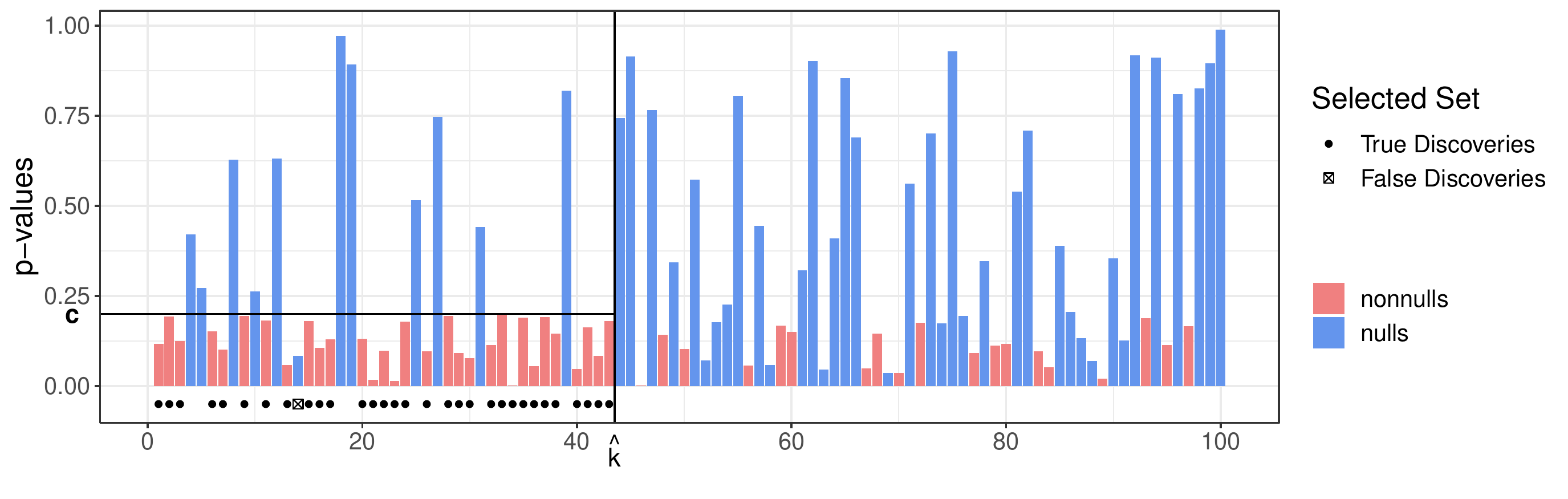}
\label{fig:bad_ordering}
\end{subfigure}
\caption{An illustration of the ordering of the $\pvals$ in Selective SeqStep+. The null $\pvals$ are sampled from $\operatorname{Unif}[0,1]$ and the nonnull $\pvals$ from $\operatorname{Unif}[0,0.2]$. We take $c = 0.2$ and $q = 0.1$. 
 The $\pvals$ are the same in the two plots, but they are ordered in two different ways. In (a), the nonnulls appear early in the sequence. In (b), the order of the $\pvals$ is random. In terms of power, Selective SeqStep+ discovers all the nonnulls in (a) but only a subset of them in (b). }
\label{fig:ordering}
\end{figure}

The second challenge concerns the ordering of the $\pvals$.  Unlike
the Benjamini-Hochberg procedure, which takes as input the $\pvals$
only, Selective SeqStep+ essentially requires two inputs: the $\pvals$
and an ordering of the $\pvals$. In other words, if we change the
order of the input $\pvals$, we could end up selecting a very
different set of variables. To illustrate this, consider the example
in Figure \ref{fig:ordering}, which fixes the $\pvals$ and compare two
different orderings. The ``good" ordering has the nonnulls appear
early in the sequence and the ``bad" ordering randomly permutes the
$\pvals$. With the good ordering, the output set contains all the
nonnulls; but with the bad ordering, only a fraction of the nonnulls is 
discovered. When the nonnull $\pvals$ appear early in the sequence,
the proportion of $\pvals$ greater than $c$ will be smaller, thus the
quantity
$\frac{1+\#\left\{j \leq k: p_{j}>c\right\}}{\#\left\{j \leq k: p_{j}
    \leq c\right\} \vee 1}$ in \eqref{eqn:seqstep} will tend to be
smaller. Therefore a larger $\hat{k}$ will be obtained and hence the
power of the procedure will be higher. In general, to make the
variable selection procedure more powerful, it is important to look
for an informative ordering that places nonnulls early in the
sequence.

Another requirement for the ordering is that it needs to be independent of the $\pvals$. FDR is in general not controlled when the $\pvals$ and the ordering are dependent. 
 As a simple example, assume a researcher obtains independent $\pvals$ and naively orders them by magnitude. Then the input sequence of $\pvals$ into Selective SeqStep+ would be an ordered sequence $p_{(1)} \leq p_{(2)} \dots \leq  p_{(p)}$. In this case, the null $\pvals$ that appear early in the sequence will tend to be smaller and hence no longer uniform. In the case of the global null (all hypotheses are null) with independent $\pvals$, we would expect to make around $c p$ false discoveries. This is because there are approximately $c p$ $\pvals$ that are smaller than or equal to $c$, and they all appear early in the sequence, hence for $k = (c + (1-c)q)p$, 
\[ \frac{1 + \#\left\{j \leq k: p_{j}>c\right\}}{\#\left\{j \leq k: p_{j} \leq c\right\} \vee 1} \approx \frac{(c + (1-c)q)p - cp}{cp} = \frac{(1-c)q}{c}.\]

In Section \ref{section:method}, we will present two methods to obtain
the ordering: the {\em split} version and the {\em symmetric
  statistic} version.  The former splits the data into two parts,
obtaining $\pvals$ from one fold, and the ordering from the
other. This makes the ordering and the $\pvals$ stochastically
independent. No data-splitting is required for the symmetric statistic
version; we obtain both the $\pvals$ and the ordering from the whole
dataset. To obtain the ordering, we compute a statistic $z_j$ for each
variable $j$, and sort the $z_j$'s. The statistic $z_j$ is obtained in
such a way that $z_j$ is marginally independent of the $\pval$ $p_j$.
In theory, this notion of independence is not sufficient for FDR
control; we however tested this method in many different empirical
settings and always controlled the $\FDR$.  In terms of power, the
symmetric statistic version is more powerful than the split
version. Thus in practice, we would recommend the symmetric statistic
version.

The third challenge is computational in nature. Recall that with the
CRT (Algorithm \ref{alg:crt}), we need to compute the test statistic
$T_j^{(b)}$ for each $j$ and each $b$. Each statistic $T_j^{(b)}$ is
obtained by sampling $\bX_j^{(b)}$ and running a machine learning
algorithm with $\bY$ as a response and $\bX_j^{(b)}, \bX_{-j}$ as
predictors. It is computationally expensive to run the machine
learning algorithm $B$ times to get a single $\pval$. In Section
\ref{section:computation}, we will present a faster way of obtaining
the test statistics and, hence, the $\pvals$.

\section{Selective SeqStep+ under dependence}
\label{section:theory}

\subsection{Almost independent $p$-values}
\label{subsection:cond_prob}
When employing SeqStep+, it is natural to ask whether the FDR is still
controlled when the $\pvals$ are ``close'' to being independent. This
section derives an bound upper bound on the FDR, which depends on the
value of
$\max_{j \in \Ho} \mathbb{P}[p_j \leq c \mid \one\cb{p_{-j} \leq
  c}]$,\footnote{For a set
  $S = \cb{j_1, \dots, j_K} \subset \cb{1, \dots, p}$, we define
  $\one\cb{p_S \leq c}$ as the Boolean vector
  $(\one\cb{p_{j_1} \leq c}, \dots, \one\cb{p_{j_K} \leq c}) $.} the
maximum of the probability that $p_j$ is at most $c$ conditional on the boolean sequence of whether other $\pvals$ are smaller than or equal to $c$. 
Under independence of the $\pvals$, it holds that
$\max_{j \in \Ho} \PP{p_j \leq c \mid \one\cb{p_{-j} \leq c}} \leq c$
since marginally, $\PP{p_j \leq c} \leq c$ for $j \in \Ho$. Our first
result states that if
$\max_{j \in \Ho} \PP{p_j \leq c \mid \one\cb{p_{-j} \leq c}}$ is
close to $c$ with high probability, then the FDR inflation cannot be
large.
\begin{theo}[Almost independent $p$-values]
\label{theo:cond_prob}
Suppose the ordering of the $\pvals$ is fixed. Set
$a_j = \PP{p_j \leq c \mid \one\cb{p_{-j} \leq c}}$ and assume the
$\pvals$ satisfy
$ \PP{ \max_{j \in \mathcal{H}_0} a_{j} \leq c + \delta} \geq 1 -
\epsilon$. Then the output from Algorithm \ref{alg:seqstep} obeys
\begin{equation}
\label{eqn:theo_cond_prob}
 \FDR \leq q \frac{c+\delta}{c}\frac{1-c}{1-c-\delta} + \epsilon. 
\end{equation}
\end{theo}

As an illustration, we describe two examples where the FDR bound can
be computed numerically. Consider data $(\bX, \bY)$, where each row of
$\bX$ is generated independently from a multivariate Gaussian
distribution with block diagonal covariance; that is, $X_j$ is only
dependent on nearby variables. Recalling that the $\pvals$ are
obtained from Algorithm \ref{alg:crt}, the first example takes the
marginal test statistic to be
$T(\bX_j, \bX_{-j}, \bY) = \abs{\Corr{\bX_j, \bY}}$, whereas in the
second example, we regress $\bY$ on $\bX_j$ and $\bX_{N(j)}$, and take
the test statistic $T(\bX_j, \bX_{-j}, \bY)$ to be the absolute value
of the fitted coefficient of $\bX_j$. Here, the elements of $N(j)$ are
the ``neighbors" of $j$, i.e. the indices in the same block as $j$.
Additional details of the simulation settings are included in Appendix
\ref{subsection:simulation_details_cond_prob}.

Before proceeding with the computation, we note that if $a_j$ were
defined conditional on additional information,
e.g.~$a_j = \PP{p_j \leq c \mid \bY, \one\cb{p_{-j} \leq c}}$, then
Theorem \ref{theo:cond_prob} would still hold. The specific block
diagonal structure of the covariance of $X$ implies that $X_i$ and
$X_j$ are independent conditionally on $Y$ if $i$ and $j$ are not in
the same block.  Thus,
$a_j = \PP{p_j \leq c \mid\bY, \one\cb{p_{-j} \leq c}} = \PP{p_j \leq
  c \mid \bY, \one\cb{p_{N(j)} \leq c}}$. For a block of size $K$, the
variable $\one\cb{p_{N(j)} \leq c}$ can take at most $2^K$ distinct
values. In practice, the conditional probability
$\PP{p_j \leq c \mid \bY, \one\cb{p_{N(j)} \leq c}}$ can therefore be
estimated using sample proportions. One can fix $\bY$, sample $\bX$
from the distribution of $X \mid Y$, compute the corresponding
$\pvals$, and compute the frequency of the event $\cb{p_j \leq c}$
conditional on the value of $\one\cb{p_{N(j)} \leq c}$. This is the
reason why the block diagonal structure of the covariance is used
here; this structure makes computations tractable since we are dealing
with $2^K$ rather than $2^{p-1}$ possible configurations.

\begin{figure}[h]
\centering
\includegraphics[width = 0.9\textwidth]{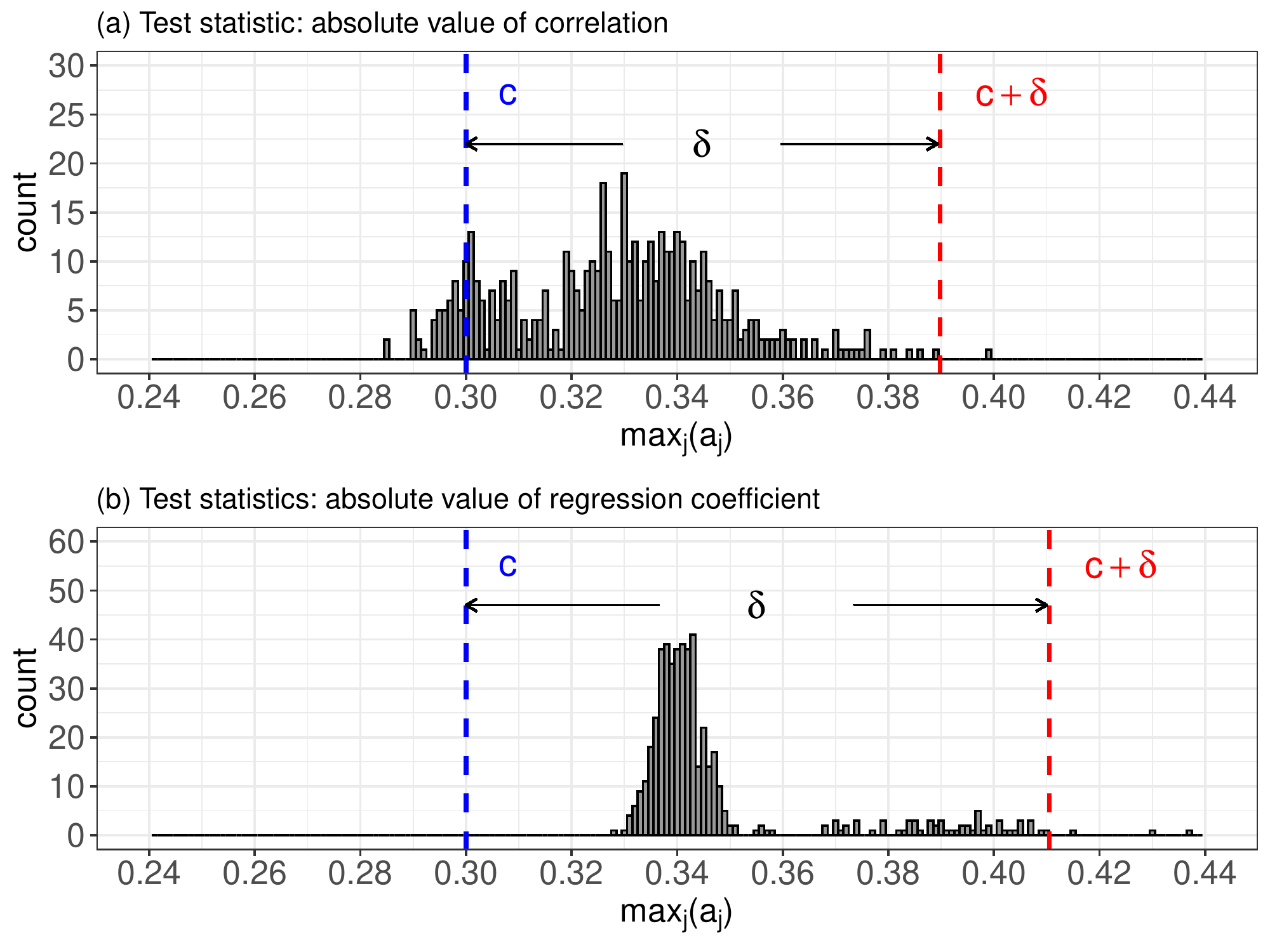}
\caption{Histogram of $\max_{j \in \mathcal{H}_0} a_{j}$ from 500
  samples. In (a), the test statistic of CRT is taken to be the
  absolute value of the correlation. In (b), the test statistic of CRT
  is taken to be the absolute value of the regression coefficient. We
  plot the threshold $c$ (blue dashed line). To apply the bound from
  Theorem \ref{theo:cond_prob}, one possible choice of $\delta$ is
  shown as the distance between the red and blue lines.}
\label{fig:hist_max_aj}
\end{figure}

In Figure \ref{fig:hist_max_aj}, we plot the histogram of
$\max_{j \in \mathcal{H}_0} a_{j} = \max_{j \in \Ho} \PP{p_j \leq c \mid \bY, \one\cb{p_{N(j)} \leq c}}$ from 500 samples and show a possible choice of $\delta$
and $\epsilon$.  Here, the FDR threshold $q$ is set to be 0.1 and $c$
is chosen to be $0.3$.  In the example where the test statistic is the
absolute value of correlation between $\bX$ and $\bY$, we can take
$\delta = 0.0893$ and $\epsilon = 0.002$. The FDR bound in
\eqref{eqn:theo_cond_prob} is thus
$q \frac{c+\delta}{c}\frac{1-c}{1-c-\delta} + \epsilon = 0.1508$. In
the other example where the test statistic is the absolute value of
the fitted regression coefficient, we can take $\delta = 0.11$ and
$\epsilon = 0.006$. The FDR bound in \eqref{eqn:theo_cond_prob} is
thus $q \frac{c+\delta}{c}\frac{1-c}{1-c-\delta} + \epsilon = 0.1682$.

\subsection{Under exchangeability}
In this section, we study whether additional structure on the $\pvals$
can be helpful in obtaining sharper FDR bounds.  To this end, consider
the assumption of exchangeability. We say that the random variables
$A_1, A_2, \dots, A_m$ are \emph{exchangeable} conditional on a random
variable $B$ if
$(A_{1}, \dots, A_{m} ) \mid B \stackrel{d}{=} (A_{\pi(1)}, \ldots,
A_{\pi(m)} ) \mid B$ for any permutation $\pi$. With this, this section makes use of the following assumption:
\begin{assu}
\label{assu:exch}
The null $\pvals$ are exchangeable conditional on the nonnull $\pvals$.
\end{assu}

To understand Assumption \ref{assu:exch}, we study examples where it
holds. Consider $\pvals$ obtained from the CRT. A sufficient set of
conditions is that the variables $X_j$'s are exchangeable and that the
test statistic $T(\cdot)$ in the CRT (Algorithm \ref{alg:crt}) is
symmetric in $\bX_{-j}$.  As a concrete example, imagine $X$ follows a
$\mathcal{N}(\mathbf{\mu}, \Sigma)$ distribution, where all entries in
$\mu$ are the same and all off-diagonal terms in $\Sigma$ are the
same. Then if $T(\bX_j, \bX_{-j}, \bY)$ is obtained by running a lasso
regression of $\bY$ on $\bX$ and taking the regression coefficient
of $j$, then the null $\pvals$ are exchangeable conditional on the
nonnulls.

Under the assumption of exchangeability, we can show that FDR
inflation will not be large. In particular, if the $\pvals$ are weakly
correlated with each other, we get a sharper upper bound.
\begin{theo}[Under exchangeability]
\label{theo:exchangeable}
Suppose the ordering of the $\pvals$ is fixed and that the nulls are
marginally stochastically larger than uniform. Under Assumption
\ref{assu:exch}, Algorithm \ref{alg:seqstep} gives
\begin{equation}
\label{eqn:exch_bound1}
\FDR \leq q + c(1-q).
\end{equation}
If, in addition, the $\pvals$ satisfy $\Corr{\one\cb{p_i \leq c}, \one\cb{p_j \leq c}} \leq \rho$ for any nulls $i \neq j$, then
\begin{equation}
\label{eqn:exch_bound2}
\FDR \leq q + \varepsilon(c, q, \rho),
\end{equation}
where
$$ \varepsilon(c, q, \rho) = \p{\frac{\delta}{1 + \beta \delta}
  \sqb{\frac{c}{1-c} - \frac{c - c\delta}{1 - (c-c\delta)}q} } \wedge
c(1-q), \quad \beta = \frac{c + (1-c)q}{(1-c)(1-q)}, \quad \delta =
\rho \frac{c(1-q) + q}{c(1-q)}.$$ The two bounds \eqref{eqn:exch_bound1} and
\eqref{eqn:exch_bound2} are sharp asymptotically. For illustration,
the bound \eqref{eqn:exch_bound2} is plotted in
Figure~\ref{fig:exch_bounds}.
\end{theo}

\begin{proof}
We will show the asymptotic sharpness of \eqref{eqn:exch_bound1} here. Specifically, we will show an example where the $\FDR$ converges to $q + c(1-q)$ as $p \to \infty$. We include a proof of the two upper bounds and the asymptotic sharpness of \eqref{eqn:exch_bound2} in Appendix \ref{subsection:proof_exch}.

Assume we are under the global null, i.e., all variables are
nulls. Set $m_0 = 1 + \lceil \frac{c p}{q + c(1-q)} \rceil$ and
consider null $\pvals$ sampled as follows:
\begin{enumerate}
\item With probability $c p/m_0$, pick $m_0$ indices uniformly at
  random from $\cb{1, \dots, p}$, and sample the corresponding
  $\pvals$ as $\text{i.i.d.}\Unif[0,c]$; sample the other $\pvals$
  independently from $\Unif[c,1]$.

\item With probability $1 - c p/m_0$, sample all $\pvals$ as
  $\text{i.i.d.}\Unif[c,1]$.
\end{enumerate}
One can easily verify that each $\pval$ marginally follows a
$\Unif[0,1]$ distribution. On the first event, we always reject all
the variables because
\[ \frac{1+\#\left\{j \leq p: p_{j}>c\right\}}{\#\left\{j \leq p:
      p_{j} \leq c\right\} \vee 1} = \frac{1 + p - m_0}{m_0} \leq
  \frac{1 + p - m_0}{m_0 - 1} \leq \frac{1-c}{c} \cdot q. \] Thus
$\FDP = 1$. On the second event, we reject none of the variables, thus
$\FDP = 0$. Combining the two cases, we get
$\FDR = cp/m_0 \to q + c(1-q)$ as $p \to \infty$.
\end{proof}

\begin{figure}
\centering
\includegraphics[width = 0.7\textwidth]{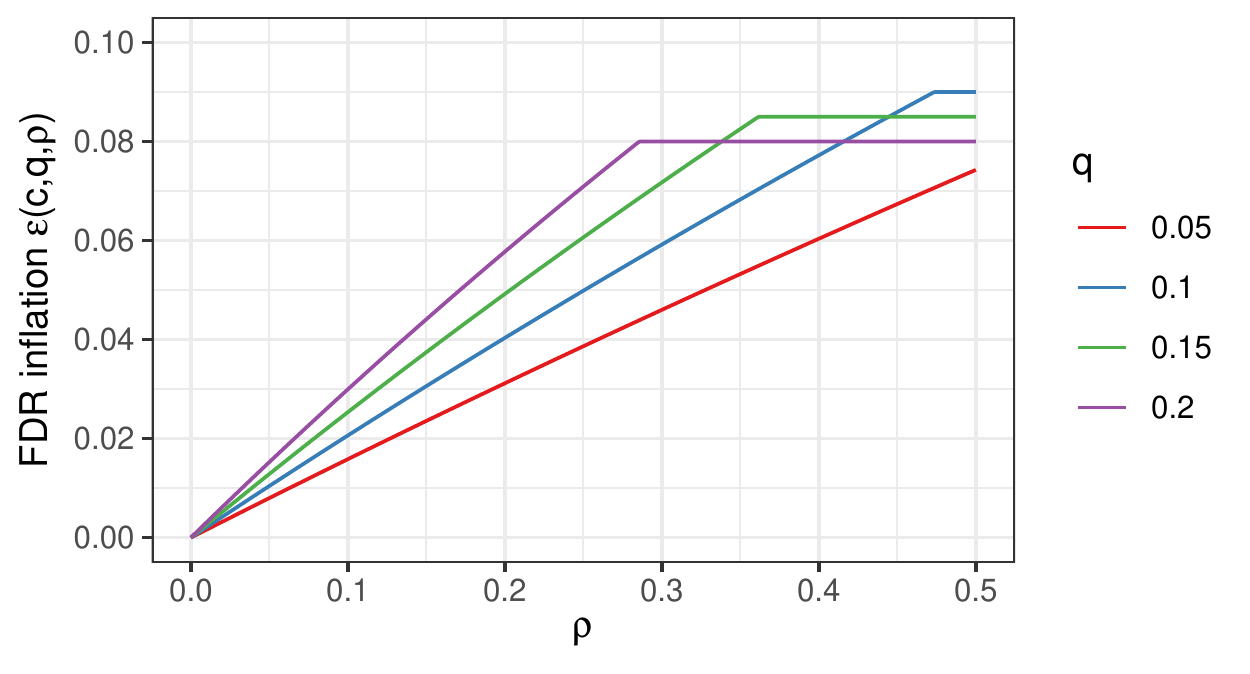}
\caption{The bound on FDR inflation $\varepsilon(c, q, \rho)$. Here we set $c= 0.1$, vary $\rho$ from 0 to 0.5, and vary the FDR threshold $q$ from $0.05$ to $0.2$. }
\label{fig:exch_bounds}
\end{figure}

Compared to the FDR bound in Theorem \ref{theo:cond_prob}, Theorem
\ref{theo:exchangeable} is neither weaker nor stronger. Theorem
\ref{theo:exchangeable} holds when the null $\pvals$ are exchangeable,
whereas Theorem \ref{theo:cond_prob} holds when the $\pvals$ are close
to being independent. When the $\pvals$ are exchangeable and highly
correlated, for example in the most extreme case where all the
$\pvals$ are the same, then \eqref{eqn:exch_bound1} in Theorem
\ref{theo:exchangeable} gives that $\FDR \leq q + c(1 - q)$, whereas
\eqref{eqn:theo_cond_prob} in Theorem \ref{theo:cond_prob} would not
be informative at all. In a different setting where the $\pvals$ are
independent but follow different distributions, Theorem
\ref{theo:cond_prob} can be used to show that $\FDR \leq q$, whereas
Theorem \ref{theo:exchangeable} cannot be applied.

\subsection{Beyond exchangeability or almost independence of the $p$-values}

In general, when the $\pvals$ have an arbitrary dependence structure,
we can bound the FDR with a logarithmic inflation; the sharpness of
the bound below is an open question.
\begin{theo}[Arbitrary dependence]
  \label{theo:no_assu}
  Suppose the ordering of the $\pvals$ is fixed and that the nulls are
  marginally stochastically larger than uniform. If $(1-c)q < c$, then Algorithm
  \ref{alg:seqstep} yields
\begin{equation}
\label{eqn:no_assu_bound}
 \FDR \leq (q + c(1-q)) \sum_{j \in \mathcal{H}_0} \frac{1}{j+1} \leq (q + c(1-q))\log p. 
\end{equation}
\end{theo}
When we have a good ordering of the $\pvals$, i.e., when the null
$\pvals$ tend to have larger indices, then the right-hand side
$(q + c(1-q)) \sum_{j \in \mathcal{H}_0} \frac{1}{j+1}$ is
smaller. Comparing to the case with exchangeability, we observe a
potential logarithmic inflation on the FDR. A similar phenomenon has
been observed for the BHq procedure, where an arbitrary dependence
among the $\pvals$ also brings a possible logarithmic inflation
\citep{benjamini2001control}.

\section{Methods to order hypotheses} 
\label{section:method}
When performing variable selection with CRT and Selective SeqStep+, it
is important to have a good ordering of the hypotheses/CRT $\pvals$. A naive way
of obtaining the ordering is as follows: apply any machine learning
algorithm to $(\bX,\bY)$, compute a statistic $z_j$ providing evidence
against the hypothesis that $j$ is null, sort the CRT $\pvals$ by
decreasing order of the $z_j$'s, and apply Selective SeqStep+.  As
argued in Section \ref{subsection:introduction_challenges}, despite
the intuitive structure of this procedure, the dependence between the
$\pvals$ and the ordering will, in general, imply a loss of FDR
control.

\subsection{Splitting}
\label{subsection:exact}

\floatname{algorithm}{Procedure}
\begin{algorithm}[h]
\caption{The Sequential CRT (Split version)}
\label{alg:exact}
\begin{algorithmic}
\REQUIRE Data $\mathcal{D} = (\bX, \bY)$, number of randomizations $B$, test statistic $T(\cdot)$, score function $Z$, $\FDR$ threshold $q$, SeqStep threshold $c$. 
\end{algorithmic}

\begin{enumerate}
\item 
\begin{algorithmic}
\STATE Split the data into two folds $\mathcal{D}_{\operatorname{pval}} = (\bX^{\operatorname{pval}}, \bY^{\operatorname{pval}})$ and $\mathcal{D}_{\operatorname{ordering}} = (\bX^{\operatorname{ordering}}, \bY^{\operatorname{ordering}})$. 
\end{algorithmic}

\item 
\begin{algorithmic}
\STATE Obtain $\pvals$ $p_1 \dots, p_p$ on $\mathcal{D}_{\operatorname{pval}}$ from CRT (Algorithm \ref{alg:crt}). 
\end{algorithmic}

\item 
\begin{algorithmic}
\STATE Compute statistics $z_j = Z\p{\bX_j^{\operatorname{ordering}},\bX_{-j}^{\operatorname{ordering}},  \bY^{\operatorname{ordering}}}$ on $\mathcal{D}_{\operatorname{ordering}}$ for each $j \in \cb{1, \dots p}$, and obtain ordering $\pi$ by sorting the statistics: $z_{\pi(1)} \geq z_{\pi(2)} \dots \geq z_{\pi(p)}$. 
\end{algorithmic}

\item 
\begin{algorithmic}
\STATE 
Apply Selective SeqStep+ (Algorithm \ref{alg:seqstep}) to $p_{\pi(1)}, p_{\pi(2)}, \dots, p_{\pi(p)}$. 
\end{algorithmic}
\end{enumerate}

\begin{algorithmic}
\ENSURE Discoveries from Selective SeqStep+. 
\end{algorithmic}

\end{algorithm}
The split version of the sequential CRT (Procedure \ref{alg:exact}) makes the $\pvals$ and
ordering independent through data splitting: the data is split into
two folds; the $\pvals$ are obtained from the CRT on the first fold;
and the ordering is obtained on the second fold. Independence ensures
that Theorem~\ref{theo:cond_prob} holds for this procedure.  The
downside is that this suffers from a power loss as is the case for
many other data splitting procedures. This motivates us to look for
procedures that use the full data to obtain both the $\pvals$ and the
ordering.

\subsection{Symmetric statistics}
\label{subsection:inexact}

As seen in Section \ref{subsection:introduction_challenges}, the correlation between the null $\pval$ $p_j$ and the statistic $z_j$, which is sorted to obtain the ordering, largely accounts for the FDR inflation. It is thus natural to seek procedures that make $p_j$ and $z_j$ independent for nulls. To this end, recall that the $\pval$ $p_j$ is defined as 
\[p_j=\frac{1}{B+1} \p{1 + \sum_{b = 1}^B \one \cb{T(\bX_j, \bX_{-j},
      \bY) \geq T(\bX_j^{(b)}, \bX_{-j}, \bY)}}. \] We propose a
method with $p_j$ as above and each $z_j$ constructed as follows:
consider a function $Z$ that is symmetric in its first $B+1$
arguments,\footnote{We say a function $h(x_1, \dots, x_n)$ is
  symmetric in its first $m$ arguments if for any permutation $\pi$ of
  $\cb{1,2,\dots, m}$,
  $h(x_1, x_2, \dots, x_m, x_{m+1}, \dots x_n) = h(x_{\pi(1)},
  x_{\pi(2)}, \dots, x_{\pi(m)}, x_{m+1}, \dots x_n)$.} and define 
\[z_j = Z\p{\bX_j, \bX_j^{(1)}, \dots, \bX_j^{(B)}, \bX_{-j}, \bY}.\]
This definition leads to Procedure \ref{alg:inexact}.

\begin{algorithm}[h]
\caption{The Sequential CRT (Symmetric statistic version)}
\label{alg:inexact}
\begin{algorithmic}
\REQUIRE Data $\mathcal{D} = (\bX, \bY)$, number of randomizations $B$, test statistic $T(\cdot)$, score function $Z(\cdot)$ (symmetric in its first $B+1$ arguments), $\FDR$ threshold $q$, SeqStep threshold $c$. 
\end{algorithmic}

\begin{enumerate}
\item 
\begin{algorithmic}
\STATE Obtain $\pvals$ $p_1 \dots, p_p$ on $\mathcal{D}$ from CRT (Algorithm \ref{alg:crt}). 
\end{algorithmic}

\item 
\begin{algorithmic}
  \STATE Compute statistics:
  $z_j = Z\p{\bX_j, \bX_j^{(1)}, \dots, \bX_j^{(B)}, \bX_{-j}, \bY}$
  for each $j \in \cb{1, \dots p}$, and obtain an ordering $\pi$ by
  sorting the statistics:
  $z_{\pi(1)} \geq z_{\pi(2)} \dots \geq z_{\pi(p)}$.
\end{algorithmic}

\item 
\begin{algorithmic}
\STATE Apply Selective SeqStep+ (Algorithm \ref{alg:seqstep}) to $p_{\pi(1)}, p_{\pi(2)}, \dots, p_{\pi(p)}$. 
\end{algorithmic}
\end{enumerate}

\begin{algorithmic}
\ENSURE Discoveries from Selective SeqStep+. 
\end{algorithmic}

\end{algorithm}

Intuitively, the $\pval$ $p_j$ is capturing the relative rank of
$\bX_j$ among $\cb{\bX_j^{(1)}, \dots, \bX_j^{(B)}}$, yet, $z_j$ is
symmetric in $\cb{\bX_j, \bX_j^{(1)}, \dots, \bX_j^{(B)}}$. The
symmetry allows us to permute elements in
$\Big\{\bX_j, \bX_j^{(1)}, \dots,$ $\bX_j^{(B)}\Big\}$ while keeping
$z_j$ fixed. This means that $z_j$ is not providing information
regarding the relative rank of $\bX_j$ among
$\cb{\bX_j^{(1)}, \dots, \bX_j^{(B)}}$ and is hence independent of
$p_j$. Put formally:
\begin{prop}
\label{theo:pj_zj}
The $\pval$ $p_j$ and the statistic $z_j$ defined in Procedure \ref{alg:inexact} obey
$p_j \independent z_j$ for any null $j$. In addition, $p_j \independent z_j |\bX_{-j}, \bY$,  for any null $j$. 
\end{prop}
This result is a special case of Proposition~\ref{prop:one_shot} that will be presented later. 

Note that Proposition \ref{theo:pj_zj} is not sufficient to guarantee
FDR control. Even though $p_j$ is independent of $z_j$, $p_j$ could,
in principle, still have a complicated relationship with
$z_{-j}$. This makes the $\pvals$ not entirely independent of the
ordering. This however does not appear to lead to FDR inflation in
practice. We indeed observe FDR control in various simulation studies
in Section \ref{section:simulation}.

The statistic $z_j$ can be computed using complicated machine learning
methods. For example, one can run a gradient boosting algorithm with regression trees as base learners, $\bY$ as a response, and
$\bX_j, \bX_j^{(1)}, \dots, \bX_j^{(B)}, \bX_{-j}$ as predictors,
obtain feature importance statistics of
$\bX_j, \bX_j^{(1)}, \dots, \bX_j^{(B)}$, and take $z_j$ to be the
maximum of the statistics. One can easily verify that with this
specific construction, $z_j$ is symmetric in
$\p{\bX_j, \bX_j^{(1)}, \dots, \bX_j^{(B)}}$.

\section{Towards faster computation: one-shot CRT}
\label{section:computation}
In the original CRT (Algorithm \ref{alg:crt}), to compute each $\pval$
$p_j$ one runs a machine learning algorithm $B$ times to obtain the
test statistics $T_j^{(b)}$ for $b \in \cb{1, \dots, B}$. This quickly
gets computationally expensive when the machine learning algorithm is
run on a large dataset. To save computation time, another way of
computing the statistics is to run the machine learning algorithm
once, with $\bY$ as a response and
$\bX_j, \bX_j^{(1)}, \dots, \bX_j^{(B)}, \bX_{-j}$ as the predictors.
Formally, we consider a procedure $T$ that takes
$\p{\bX_j^{(0)}, \bX_j^{(1)}, \dots, \bX_j^{(B)}, \bX_{-j}, \bY}$ as
input, and outputs importance statistics $T_j^{(b)}$ for each
$b \in \cb{1, \dots, B}$, i.e.,
\begin{equation}
\label{eqn:one_shot_procedure}
 \p{T_j^{(0)}, \dots, T_j^{(B)}} = T\p{\bX_j^{(0)}, \bX_j^{(1)}, \dots, \bX_j^{(B)}, \bX_{-j}, \bY}.  
\end{equation}
We restrict attention to procedures obeying the following symmetry
property: for all permutations $\pi$,
\begin{equation}
\label{eqn:one_shot_symmetric}
 T\p{ \sqb{\bX_j^{(0)}, \bX_j^{(1)}, \dots, \bX_j^{(B)} }_{\operatorname{perm}(\pi)} , \bX_{-j}, \bY  } 
 =  \sqb{T\p{\bX_j^{(0)}, \bX_j^{(1)}, \dots, \bX_j^{(B)}  , \bX_{-j}, \bY  } }_{\operatorname{perm}(\pi)}. 
\end{equation}
This is saying that if we permute the input $\bX_j^{(b)}$, this has the effect of permuting the output statistics. 
With these statistics, we obtain $\pvals$ via 
\begin{equation}
\label{eqn:one_shot_pvalue}
p_j=\frac{1}{B+1}
	\p{1 + \sum_{b = 1}^B \one{\cb{T_j^{(0)} \leq T_j^{(b)} }}}.  \footnote{We break ties randomly. }
\end{equation}
 We call this procedure \textit{one-shot CRT}. 

 As a concrete example, consider a case where the lasso is used to
 compute the test statistic $T_j^{(b)}$. To compute each $\pval$
 $p_j$, the original CRT runs the lasso by regressing $\bY$ on
 $\bX_j^{(b)}, \bX_{-j}$ for each $b \in \cb{0, 1, \dots, B}$, and
 takes $T_j^{(b)}$ to be the absolute value of the fitted coefficient
 for $\bX_j^{(b)}$. In total, we run $B+1$ regressions. In contrast, the one-shot
 CRT runs the lasso only once by regressing $\bY$ on
 $\bX_j^{(0)}, \bX_j^{(1)}, \dots, \bX_j^{(b)} , \bX_{-j}$ and takes
 $T_j^{(b)}$ to be the corresponding $|\hat{\beta}|$ for
 $\bX_j^{(b)}$ (this obeys \eqref{eqn:one_shot_symmetric}).

The symmetry in \eqref{eqn:one_shot_symmetric} ensures that the theoretical properties of the CRT $\pvals$ still hold for the one-shot CRT $\pvals$.
\begin{prop}
\label{prop:one_shot}
Consider a null variable $j$. Assume that the $\pval$ $p_j$ is obtained from \eqref{eqn:one_shot_procedure} and \eqref{eqn:one_shot_pvalue}, and that \eqref{eqn:one_shot_symmetric} holds. Then $p_j$ satisfies $\PP{p_j \leq \alpha} \leq \alpha$, for any $\alpha \in [0,1]$, and $p_j \independent z_j|\bX_{-j}, \bY$, where $z_j$ is defined in Procedure \ref{alg:inexact}.
\end{prop}
\begin{proof}
  For the sake of notation, set $\bX_j^{(0)} = \bX_j$. Consider a null
  $j$. By construction of $\bX_j^{(b)}$, all the $\bX_j^{(b)}$'s are
  i.i.d.~conditional on $\bX_{-j}$ and $\bY$. Thus for any permutation
  $\rho$ of $\cb{0, \dots, B}$,
  $\p{\bX_j^{(0)}, \dots, \bX_j^{(B)}} \stackrel{d}{=}
  \p{\bX_j^{(\rho(0))}, \dots, \bX_j^{(\rho(B))}} \Big| \bX_{-j},
  \bY$. The symmetry of $Z$ in its first $B+1$ arguments further
  ensures that
  $z_j = Z\p{\bX_j^{(0)}, \bX_j^{(1)}, \dots, \bX_j^{(B)}, \bX_{-j},
    y} = Z\p{\bX_j^{(\rho(0))}, \bX_j^{(\rho(1))}, \dots,
    \bX_j^{(\rho(B))}, \bX_{-j}, \bY}$. Combining these facts, we have
\[ \p{\bX_j^{(0)}, \dots, \bX_j^{(B)}} \stackrel{d}{=} \p{\bX_j^{(\rho(0))}, \dots, \bX_j^{(\rho(B))}} \Big| \bX_{-j}, \bY, z_j.\]
By property \eqref{eqn:one_shot_symmetric}, 
\[ \p{T_j^{(\rho(0))}, \dots, T_j^{(\rho(B))}} =
  T\p{\bX_j^{(\rho(0))}, \bX_j^{(\rho(1))}, \dots, \bX_j^{(\rho(B))},
    \bX_{-j}, \bY}. \] This term has the same distribution as
$T\p{\bX_j^{(0)}, \bX_j^{(1)}, \dots, \bX_j^{(B)}, \bX_{-j}, \bY}$
conditional on $\bX_{-j}, \bY$ and $z_j$.  Since
$T\p{\bX_j^{(0)}, \bX_j^{(1)}, \dots, \bX_j^{(B)}, \bX_{-j}, \bY}=
\p{T_j^{(0)}, \dots, T_j^{(B)}}$, we have
\[ \p{T_j^{(\rho(0))}, \dots, T_j^{(\rho(B))}} \stackrel{d}{=}
  \p{T_j^{(0)}, \dots, T_j^{(B)}} \Big| \bX_{-j}, \bY, z_j.\] This
implies that conditional on $\bX_j, \bY, z_j$,
$p_j \sim \operatorname{Unif}\cb{\frac{1}{B+1}, \frac{2}{B+1}, \dots,
  1}$. Note that the same holds without conditioning on $z_j$, i.e.,
conditioning on $\bX_j, \bY$ only. Hence
$p_j \independent z_j | \bX_{-j}, \bY$. Finally, the claim that
$\PP{p_j \leq \alpha} \leq \alpha$ follows from the fact that the
distribution of
$\operatorname{Unif}\cb{\frac{1}{B+1}, \frac{2}{B+1}, \dots, 1}$ is
stochastically greater than $\operatorname{Unif}[0,1]$.
\end{proof}
In terms of FDR control, the theorems from Section
\ref{section:theory} still hold for either the split or symmetric statistic
version of our variable selection method applied with the one-shot CRT
$\pvals$.

As an illustration, we show the average computation time of the
one-shot CRT and the original CRT (on all variables together) on
synthetic datasets in Table \ref{table:time}. The number of
randomizations is set to $B = 9$, and other details of the simulation
study are included in Appendix
\ref{subsection:simulation_details_table}. Compared with the original
CRT, the one-shot CRT reduces the computation time by a factor of
roughly $1/B$ as expected.
\begin{table}[h]
\centering
\caption{Average computation times (seconds)}
\label{table:time}
\begin{tabular}{|p{2.8cm} |p{2.8cm}| p{2.8cm}| p{2.95cm}| p{2.95cm}|}
\hline
Setting   & Linear & Logistic & Non-linear 1 & Non-linear 2 \\ 
\hline
Dimension &$n = 300, p = 300$ &$n = 300, p = 300$ &$n = 500, p = 200$ & $n = 500, p = 200$     \\ 
\hline
Statistics are  computed with& lasso  & glmnet & gradient boosting  & gradient boosting  \\
\hhline{|=|=|=|=|=|}
Original CRT   & 694 & 709          & 2811       & 2672    \\
 \hline
One-shot CRT     & 89  & 91             & 248         & 277     \\
 \hline
\end{tabular}
\end{table}

Adding $B$ predictors into a machine learning
algorithm cannot be used if we are combining CRT with BHq. As BHq
requires much finer $\pvals$, $B$ needs to be much larger. Adding many
irrelevant predictors into a regression problem is not generally a
wise move.

More generally, the computational problem posed by the combination of
the CRT and BHq has been considered by \citet{tansey2018holdout} and
\citet{liu2020fast}, which propose separate methods to reduce the
running time.  \citet{tansey2018holdout} consider data splitting: the
algorithm trains a complicated machine learning model on the first
part of the data and obtains $\pvals$ on the second part of the data
making use of the trained model. The data splitting trick ensures that
the complicated machine learning model will be fitted only once, and
hence makes the algorithm much faster and computationally
feasible. \citet{liu2020fast} propose a technique called
\emph{distillation}. Their proposed algorithm distills all the
high-dimensional information in $\bX_{-j}$ about $\bY$ into a
low-dimensional representation, computes the test statistic as a
function of $\bX_j$, $\bY$, and the low-dimensional representation,
and obtains the $\pvals$ based on the test statistics. The computation
time is much lower since the expensive model fitting takes place in
the distillation step, which is performed only once for each $j$.  The
two methods both give marginally valid $\pvals$, but the $\pvals$
would not be independent in general, and there is no theoretical
guarantee on FDR control. (Both papers confirm in their simulations
that the FDR of each method is well controlled empirically.)

\section{Simulations}
\label{section:simulation}

In this section, we demonstrate the performance of our methods on synthetic data.
Software for our method is available from \url{https://github.com/lsn235711/sequential-CRT}, along with code to reproduce the analyses. We include in Appendix \ref{section:simulation_details} implementation details and additional simulation studies. 

\subsection{Comparison of the original CRT and the one-shot CRT}
\label{subsection:simulation_small}
We compare the proposed symmetric statistic version of the sequential CRT (with one-shot CRT) and the sequential CRT (with the original CRT). We also
compare our methods with Model-X knockoffs as a benchmark.  We
consider a few different settings: linear/non-linear(tree like)
models, and Gaussian/binomial responses. In all settings, the number of true nonnulls is set to be 20. For the distribution of $X$,
we consider a Gaussian autoregressive model and a hidden Markov
model. To compute test statistics, we consider algorithms including
$L_1$-regularized regression (glmnet) and gradient boosting with
regression trees as base learners. The statistic $z_j$ in Procedure
\ref{alg:inexact} is taken to be
$z_j = \max_{b \in \cb{0, \dots, B}}T_j^{(b)}$, where $T_j^{(b)}$ is
the test statistic computed in the CRT. Knockoffs are constructed with
the Gaussian semi-definite optimization algorithm
\citep{candes2018panning} for the Gaussian autoregressive model, and
with Algorithm 3 from \citep{sesia2019gene} for the hidden Markov
model. Details of the simulation study are included in Appendix
\ref{subsection:simulation_details_small}. Figure
\ref{fig:fdr_power_small} compares the performance of the above
methods in terms of empirical false discovery rate and power averaged
over 100 independent replications. In all settings, the sequential CRT appears
to control the FDR around the desired level $q = 0.1$.  In terms of
power, the performance of the one-shot CRT appears to be similar to
that of the original CRT in most of the settings.  Compared to
knockoffs, the sequential CRT (both original CRT and
one-shot CRT) is more powerful.

\begin{figure}
\begin{subfigure}{\textwidth}
\caption{$X$ follows a Gaussian AR model}
\includegraphics[width = \textwidth]{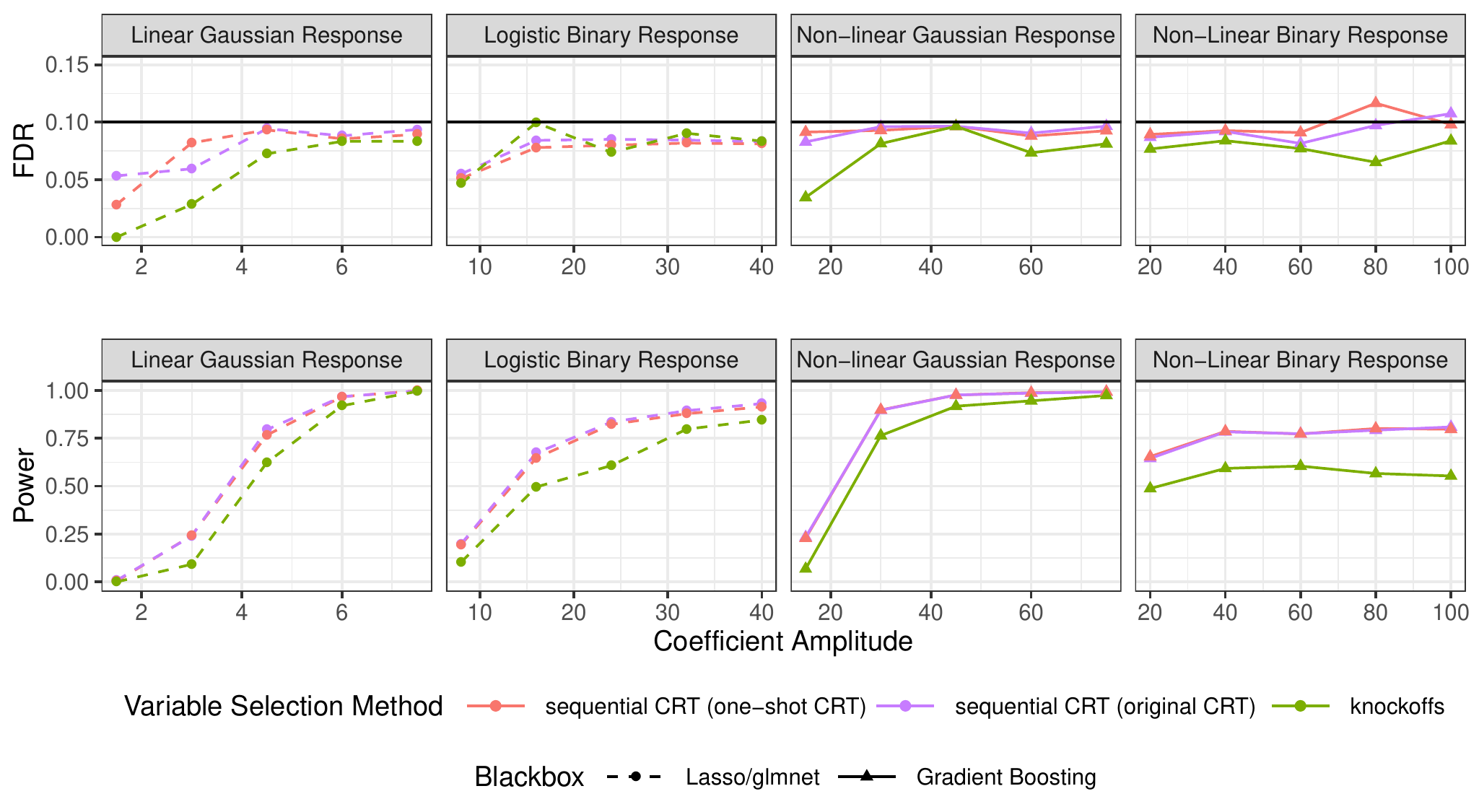}

\end{subfigure}
\noindent\rule{\textwidth}{0.4pt}
\begin{subfigure}{\textwidth}
\caption{$X$ follows an HMM}
\includegraphics[width = \textwidth]{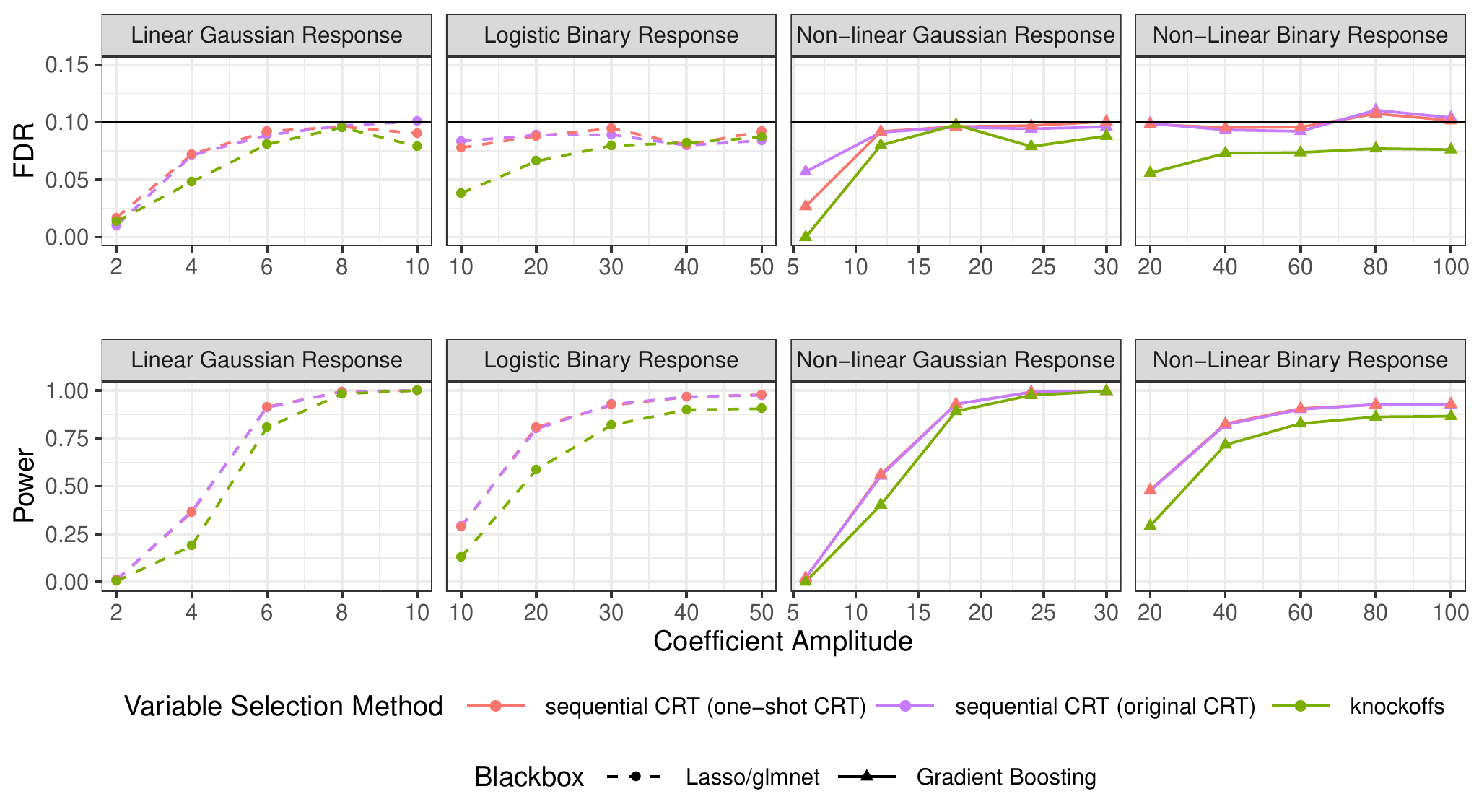}
\end{subfigure}
\caption{Performance of the sequential CRT with the original CRT and one-shot CRT (both use the symmetric statistic) and knockoffs on small synthetic datasets. The nominal false discovery rate level is 10\%. Results are averaged over 100 independent experiments.}
\label{fig:fdr_power_small}
\end{figure}

\subsection{Comparison of the sequential CRT with knockoffs}
\label{subsection:simulation_large}
 We compare the proposed split version (Procedure \ref{alg:exact}) and symmetric statistic version (Procedure \ref{alg:inexact}) of the sequential CRT with Model-X knockoffs. We run the sequential CRT with one-shot CRT. 
We consider similar settings as in the above Section \ref{subsection:simulation_small}. Since the computation time of one-shot CRT is much lower compared to the original CRT, here we run the experiments on larger datasets. In all settings in this section, the number of nonnulls is set to be 50.
Other details can be found in Section \ref{subsection:simulation_small} and Appendix \ref{subsection:simulation_details_large}. Figure \ref{fig:fdr_power_large} compares the performance of the above methods in terms of empirical false discovery rate and power averaged over 100 independent replications. In all settings, the sequential CRT appears to control the FDR around the desired level $q = 0.1$. In terms of power, the symmetric statistic version is comparable to knockoffs. 

\begin{figure}
\begin{subfigure}{\textwidth}
\caption{$X$ follows a Gaussian AR model}
\includegraphics[width = \textwidth]{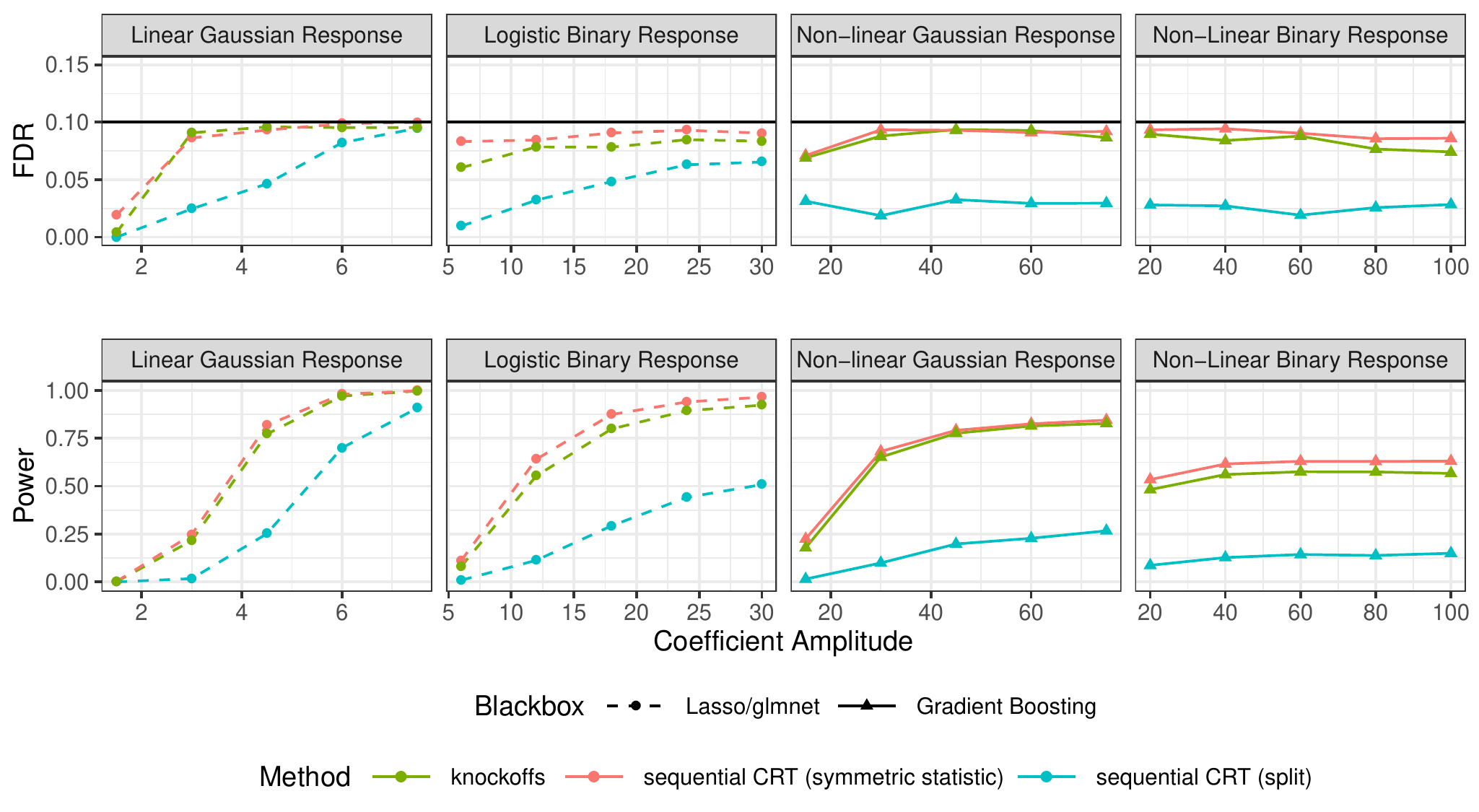}

\end{subfigure}
\noindent\rule{\textwidth}{0.4pt}
\begin{subfigure}{\textwidth}
\caption{$X$ follows an HMM}
\includegraphics[width = \textwidth]{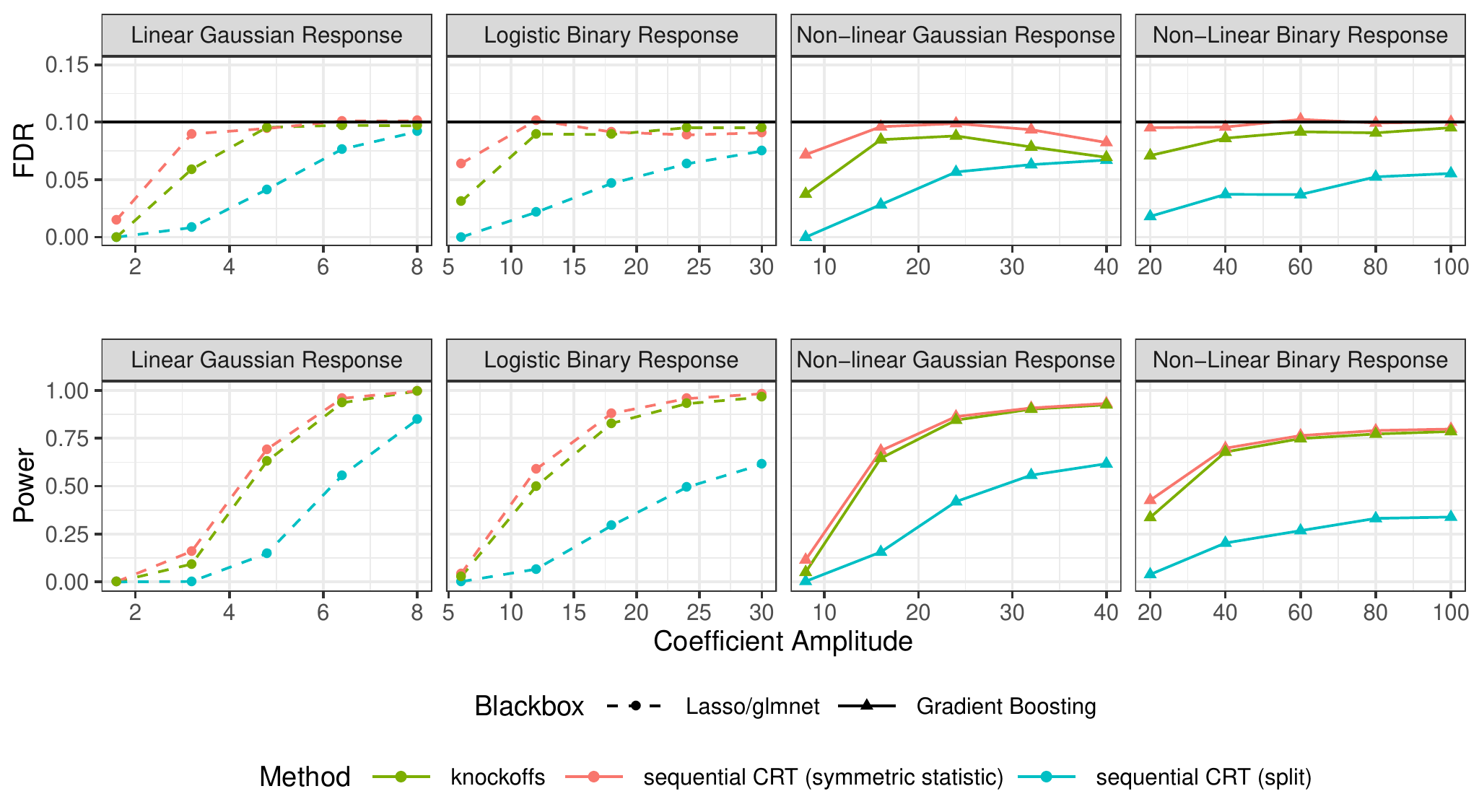}
\end{subfigure}
\caption{Performance of the proposed split version and symmetric statistic version of the sequential CRT compared to knockoffs on larger synthetic datasets. The nominal false discovery rate level is 10\%. Results are averaged over 100 independent experiments.}
\label{fig:fdr_power_large}
\end{figure}

\subsection{The role of the number of potential discoveries}
\label{subsection:simulation_num_nonnull}
\begin{figure}
\includegraphics[width = \textwidth]{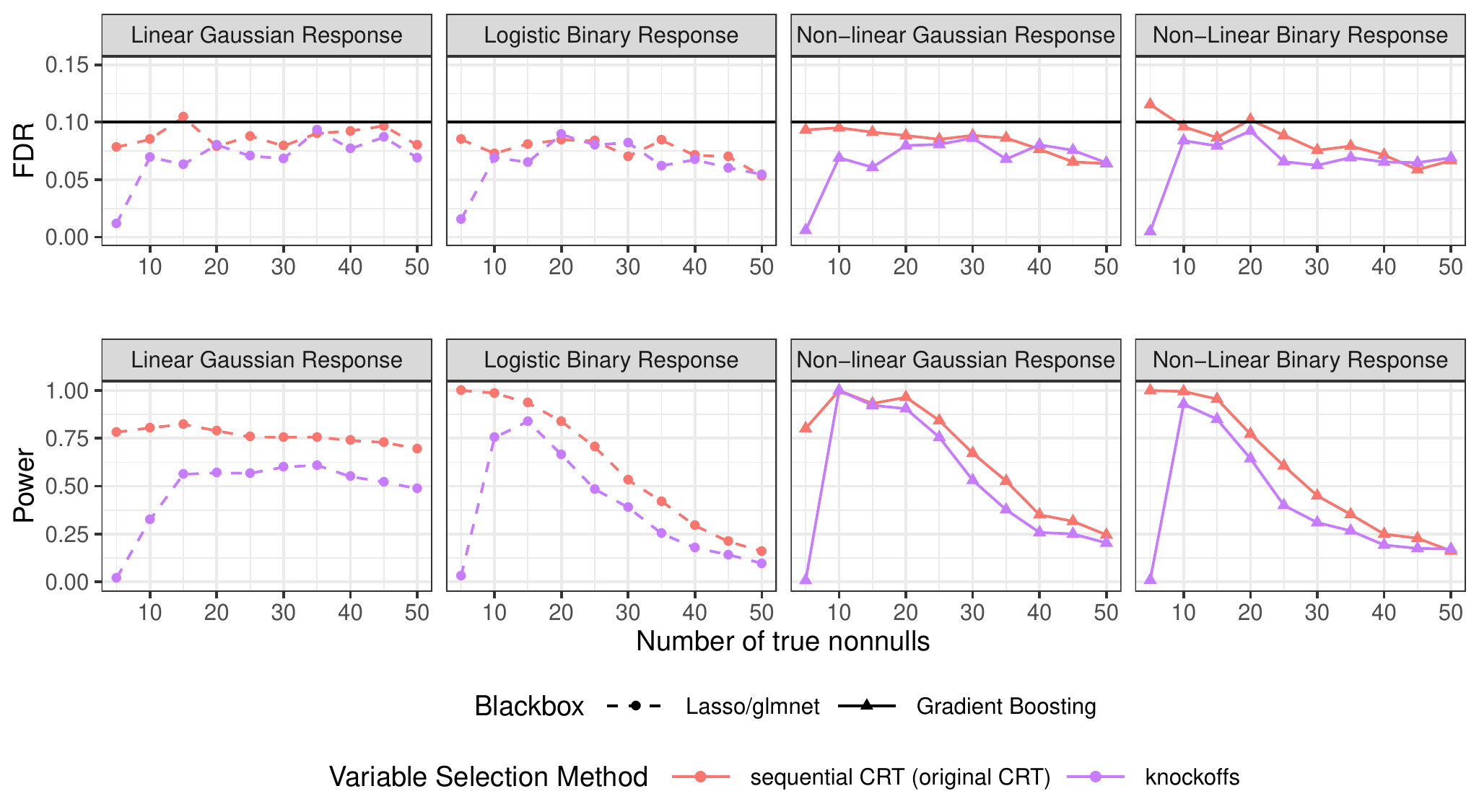}
\caption{Performance of the sequential CRT compared to knockoffs as a
  function of the number of non-nulls. The nominal false discovery
  rate level is 10\%. Results are averaged over 100 independent
  experiments.}
\label{fig:fdr_power_num_nonnull}
\end{figure}

Comparing the results from Sections \ref{subsection:simulation_small}
and \ref{subsection:simulation_large}, we observe that the power gain
of the sequential CRT vis-a-vis model-X knockoffs is more noticeable
when the number of nonnulls is small.  To understand this phenomenon,
return to the connection between the knockoff filter and Selective
SeqStep+. It was shown by \citet{barber2015controlling} that the
knockoff filter can be cast as a special case of the Selective
SeqStep+ applied to ``one-bit'' p-values with $c$ chosen to be 0.5.
When $q = 0.1$ and $c = 0.5$, the selected set of the Selective
SeqStep+ becomes
$\hat{\mathcal{S}} = \cb{ j \leq \hat{k}: p_{j} \leq c}$, where
\begin{equation}
\label{eqn:seqstep_knockoffs}
\hat{k} =\max \left\{k \in \cb{1,\dots, p}: \frac{1+\#\left\{j \leq k: p_{j}>0.5\right\}}{\#\left\{j \leq k: p_{j} \leq 0.5\right\} \vee 1} \leq 0.1 \right\}. 
\end{equation}
When the number of nonnulls is small, the set above may be
empty. Consider an example where the number of nonnulls is 8. Even in
the ideal case where all the nonnulls have vanishing $\pvals$ and the
nonnulls appear early in the sequence, for any $k \geq 8$, the left
hand side in the inequality \eqref{eqn:seqstep_knockoffs} becomes
\begin{equation}
\frac{1 + \# \cb{ \text{null } 9\leq j\leq k: p_j \geq 0.5}}{8 + \# \cb{ \text{null } 9\leq j \leq k: p_j < 0.5}} \gtrsim \frac{1}{8} > 0.1.
\end{equation}
Therefore, most of the time there is no $k$ satisfying the inequality \eqref{eqn:seqstep_knockoffs}, thus we make no rejections. 
Hence the power will be low. 

The sequential CRT, however, will not suffer from the same problem. We recall that throughout this paper, we take $c=0.1$ in the sequential CRT. With $c = 0.1$, the definition of $\hat{k}$ becomes
\begin{equation}
\hat{k} =\max \left\{k \in \cb{1,\dots, p}: \frac{1+\#\left\{j \leq k: p_{j}>0.1\right\}}{\#\left\{j \leq k: p_{j} \leq 0.1\right\} \vee 1} \leq 0.9 \right\}. 
\end{equation}
With a good ordering the $\pvals$, the left hand side of the
inequality can easily become lower than 0.9, a much less stringent
threshold.

We run simulations varying the number of nonnulls.  We compare
the proposed symmetric statistic version of the sequential CRT with
model-X knockoffs. We run the sequential CRT with one-shot CRT.  We
consider settings as in Section
\ref{subsection:simulation_small}; details are  in Appendix
\ref{subsection:simulation_details_num_nonnull}. Figure
\ref{fig:fdr_power_num_nonnull} compares the performance of the above
methods in terms of empirical false discovery rate and power averaged
over 100 independent replications. In all settings, the sequential CRT
appears to control the FDR around the desired level $q = 0.1$. In
terms of power, we see that the sequential CRT overcomes ``the threshold
phenomenon'' discussed earlier. 

\subsection{Choice of the threshold $c$ and the number $B$ of randomizations}
\label{subsection:simulation_Bc}
We here study the effect on power of the threshold $c$ and of the
  number $B$ of randomizations. We focus on the sequential CRT
  (symmetric statistics version with one-shot CRT). Intuitively, we
  expect the procedure with a smaller $c$ and a smaller $B$ to be more
  powerful. With a smaller $c$, our procedure is more likely to overcome ``the threshold phenomenon'' as discussed in Section \ref{subsection:simulation_num_nonnull}. When using a smaller value of $B$, we make sure
  that we are not including too many irrelevant predictors in the
  machine learning algorithm while running the one-shot CRT
  algorithm. In addition, when we compute the statistics $z_j$ in
  Algorithm \ref{alg:inexact}, we take the maximum (or the difference
  between the maximum and the median) of the feature importance
  statistics of $\bX_j, \bX_j^{(1)}, \dots, \bX_j^{(B)}$; thus it is
  helpful to have a smaller $B$ so that the signal, i.e., the feature
  importance statistics of $\bX_j$ has a chance of standing out. If we
  use an extremely large value of $B$, there is a chance that the
  maximum of the feature importance statistics of
  $\bX_j^{(1)}, \dots, \bX_j^{(B)}$ exceeds that of $\bX_j$.  That
  said, $c$ and $B$ cannot be too small at the same time. At the very
  least, in order for our procedure to make any rejection, we need to
  have some $\pvals$ no larger than $c$. Since the $\pvals$ are
  bounded below by $1/(B+1)$, a necessary condition for not being
  powerless is to have $c \geq 1/(B+1)$.

\begin{figure}
\centering
\begin{subfigure}[b]{0.9\textwidth}
      \caption{Small datasets}
         \centering
         \includegraphics[width = \textwidth]{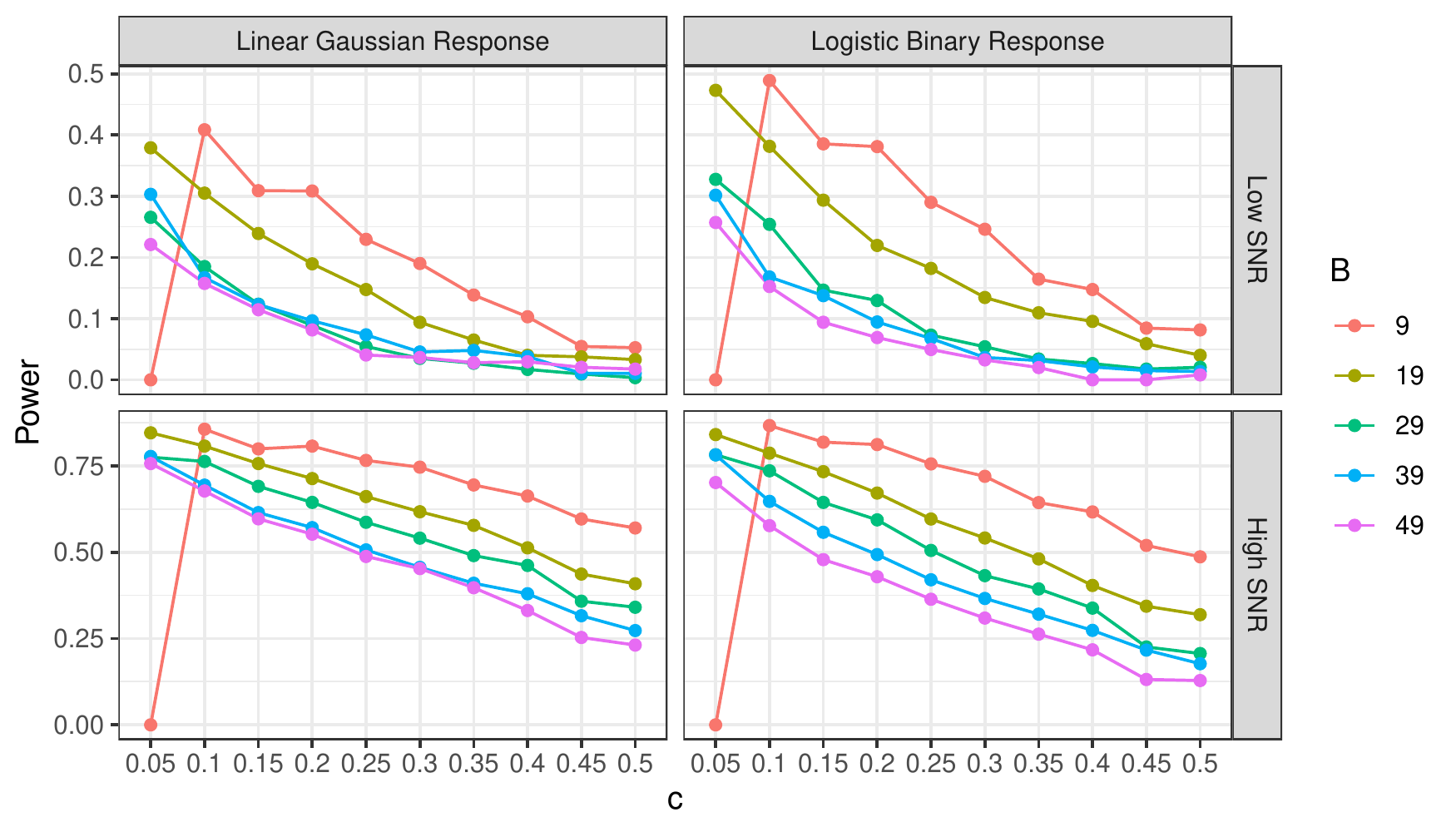}
         \label{fig:Bc_small}
 \end{subfigure}
\begin{subfigure}[b]{0.9\textwidth}
\noindent\rule{\textwidth}{0.4pt}
  \caption{Large datasets}
\centering
 \includegraphics[width = \textwidth]{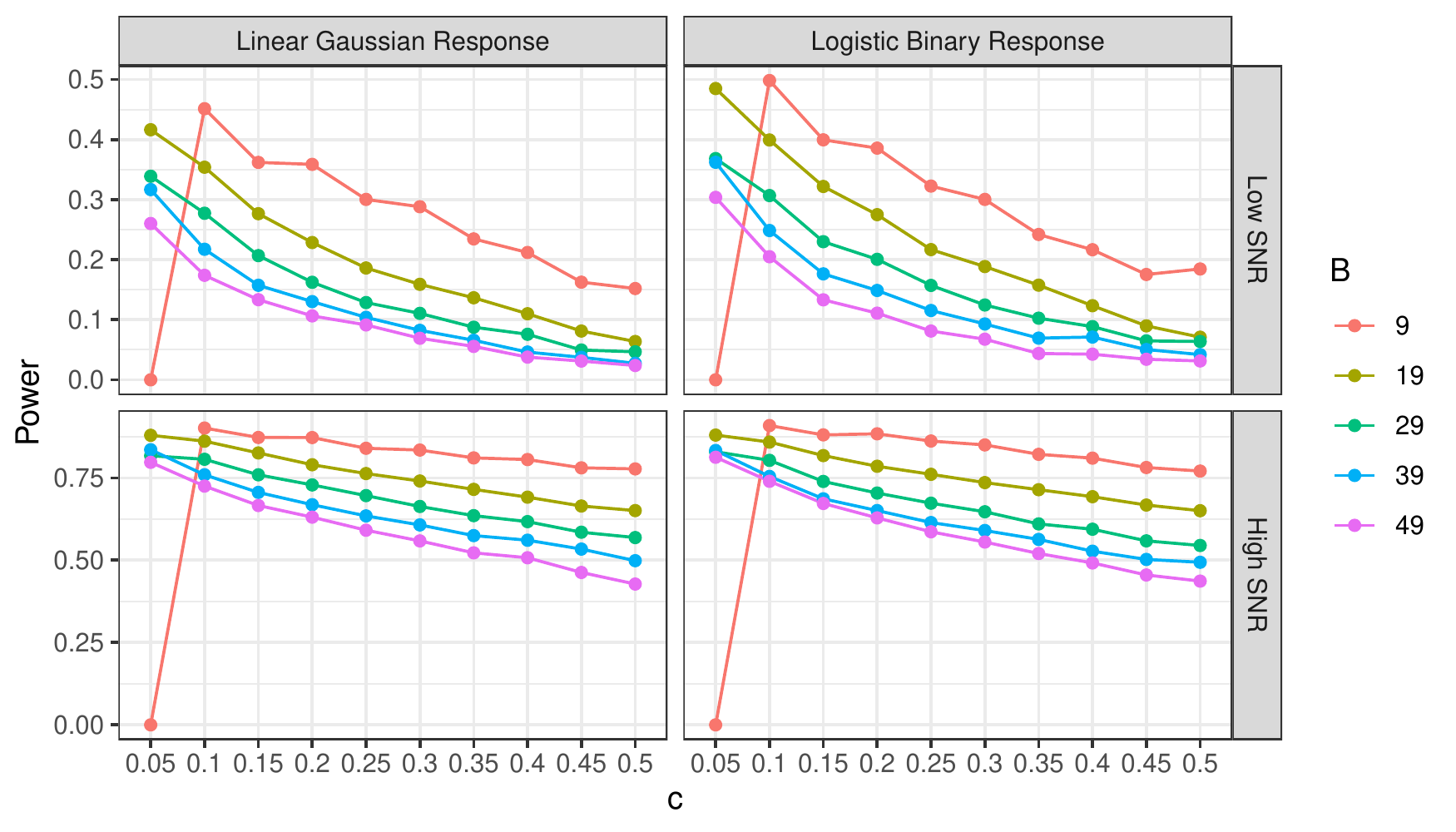}
 \label{fig:Bc_large}
\end{subfigure}
\caption{Power of the sequential CRT (symmetric statistics version with one-shot CRT) with different choices of threshold $c$ and number of randomizations $B$. The nominal false discovery
  rate level is 10\%. Results are averaged over 100 independent
  experiments.}
\label{fig:Bc}
\end{figure}

The above heuristic arguments are confirmed in simulation
studies. Figure \ref{fig:Bc} compares the performance of the
sequential CRT with different values of $c$ and $B$.\footnote{In our
  simulation studies, we take $B = 10k + 9$ for some
  $k \in \mathbb{Z}$ because we want to make $B+1$ a multiple of 10,
  and thus make it possible for $\ell/(B+1) = c$ to hold for some
  integer $\ell$. } We consider several settings: linear/logistic
models, small/large synthetic datasets, low/high signal to noise
ratios. We observe the same phenomenon in all settings. Namely, power
increases as $B$ decreases and $c$ decreases with the caveat that
they cannot both be small at the same time. It appears that the pair
$(B, c) = (9, 0.1)$ is the most powerful in all settings, justifying
the choices we made in earlier simulation studies. In Appendix
\ref{subsection:simulation_details_Bc}, we provide implementation
details to reproduce Figure \ref{fig:Bc} and additionally show that
the FDR is controlled at the nominal level for all choices of $B$ and
$c$.

\section{Real data application}
\label{section:application}
We now apply our method to a breast cancer dataset to identify gene expressions on which the cancer stage depends. The
dataset is from \cite{curtis2012genomic}, which consists of
$n = 1,396$ staged cases of breast cancer. For each case, the data
consists of expression level (mRNA) and copy number aberration (CNA)
of $p = 164$ genes. The goal is to identify genes whose expression
level is not independent of the cancer stage, conditioning on all
other genes and CNAs.  The response variable, the progression stage of
breast cancer, is binary.  We take the dataset from
\citep{liu2020fast} and pre-process the data as in their work. We
refer to Section~5 and Section~E of \citep{liu2020fast} for further
details. Following \citep{liu2020fast}, we model the distribution of
expression levels using a multivariate Gaussian.  The nominal false
discovery rate is set to be 10\%.

Below, we compare the following  methods:
\begin{enumerate}
\item \emph{Sequential CRT}: We consider Procedure \ref{alg:inexact} with one-shot CRT. We take the SeqStep threshold $c$ to be 0.1, and the number of randomizations to be 9. We take the importance statistics to be the absolute values of the coefficient of a cross-validated $L_1$-penalized logistic regression. 
\item \emph{Distilled CRT} \citep{liu2020fast}. We consider both $\operatorname{d}_0$CRT and $\operatorname{d}_{\operatorname{I}}$CRT; we refer to Section 2.3 and 2.4 of \citep{liu2020fast} for specific constructions of the dCRT. Since the response variable is binary, the distillation step is done by a cross-validated $L_1$-penalized logistic regression. 
\item \emph{HRT} \citep{tansey2018holdout}. Algorithm 1 of \citep{tansey2018holdout} is implemented with a logistic model fitted by a cross-validated $L_1$-penalized logistic regression and a data split of 50\%-50\%. 
\item \emph{Knockoffs} \citep{candes2018panning}. Knockoffs are constructed with the Gaussian semi-definite optimization algorithm. We take the
feature importance statistic to be the glmnet coefficient difference. 
\end{enumerate}
For \emph{Distilled CRT} and \emph{HRT}, we reproduce the analysis from \citet{liu2020fast}. All methods considered above are randomized procedures, i.e., different runs of the same algorithm produce possibly different sets of discoveries. We run each method 100 times and compare the number of discoveries.
 Figure \ref{fig:real_data_num_disc} is a boxplot showing the number of discoveries across random seeds. Our procedure appears to make more discoveries on average than the other methods. Compared to knockoffs, our procedure has less variability. 

\begin{figure}
\centering
\includegraphics[width = 0.8 \textwidth]{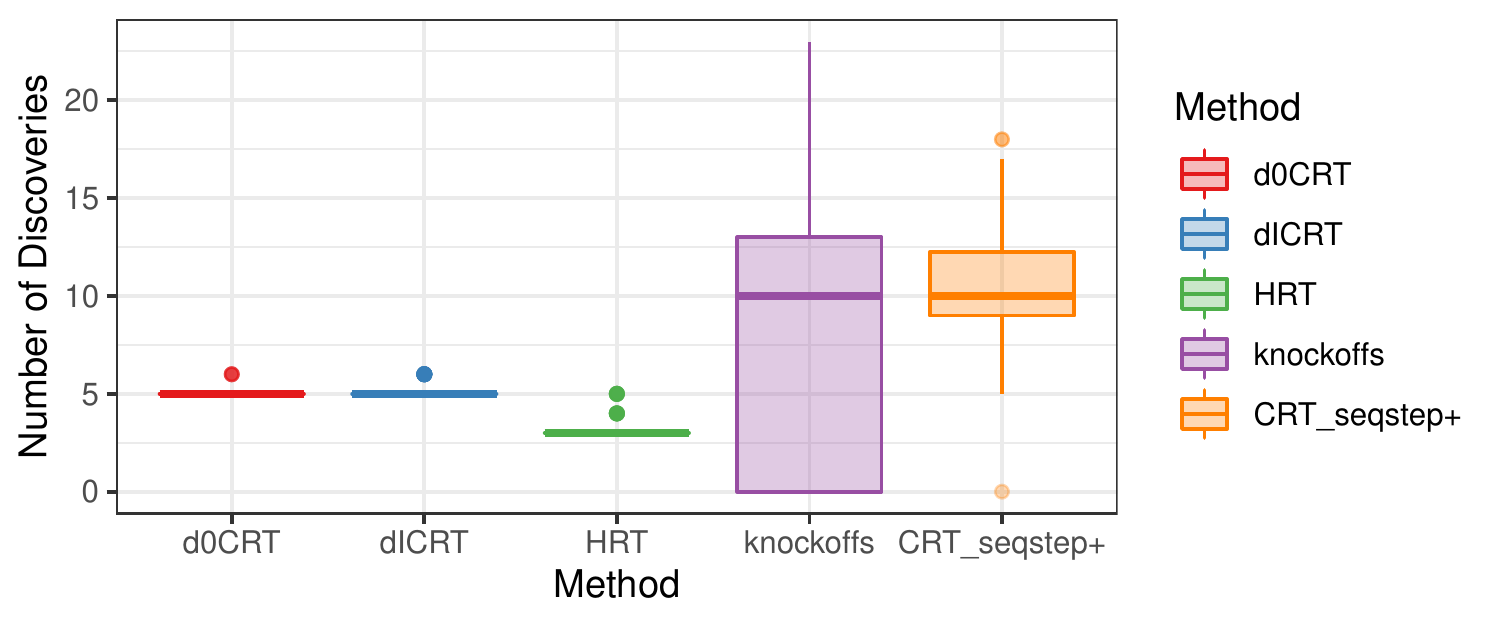}
\caption{Number of discoveries on the breast cancer dataset}
\label{fig:real_data_num_disc}
\end{figure}

We present the full list of genes discovered by the sequential CRT in Table
\ref{table:discoveries_validate}. 
Note here that the sequential CRT is random, thus different runs could produce different results. Following \cite{candes2018panning} and \cite{sesia2019gene}, to make the discoveries more ``reliable", we run the proposed method multiple times and only show the genes that are
selected by our procedure more than 10\% of the time. The 10\% level is somewhat arbitrary and we do not make any claims about the discoveries exceeding this threshold. We leave the study of setting a threshold achieving theoretical error control guarantees to future research. We however observe that all discoveries above the 10\% threshold were shown in other independent studies to be related to the development of cancer.

\begin{table}
\centering
\begin{tabular}{ |c |c |c|c|} 
\hline
Selection frequency & Gene &  Discovered by dCRT? &Confirmed in?  \\
\hline
99\% & {\em HRAS} & Yes&  \citet{geyer2018recurrent}\\
\hline
99\% & {\em RUNX1} & Yes & \citet{li2019runx1}\\
\hline
96\% & {\em FBXW7} & Yes &  \citet{liu2019fbxw7}\\ 
\hline
95\% & {\em GPS2} & Yes & \citet{huang2016g}\\ 
\hline
95\% & {\em NRAS} & & \citet{galie2019ras}  \\ 
\hline
82\% & {\em FANCD2} & & \citet{rudland2010significance}  \\ 
\hline
78\% & {\em MAP3K13} &Yes & \citet{han2016microrna}\\ 
\hline
76\% & {\em AHNAK} &  & \citet{chen2017ahnak}\\ 
\hline
67\% & {\em MAP2K4} &  &\citet{liu2019map2k4}\\ 
\hline
58\% & {\em CTNNA1} &  &\citet{clark2020loss} \\ 
\hline
35\% & {\em NCOA3} & & \citet{gupta2016ncoa3}\\ 
\hline
13\% & {\em LAMA2} & & \citet{liang2018targeted}\\ 
\hline
10\% & {\em GATA3} & & \citet{mehra2005identification}\\ 
\hline
\end{tabular}
\caption{Discoveries made by Procedure \ref{alg:inexact} on the breast cancer dataset}
\label{table:discoveries_validate}
\end{table}

\section{Discussion}
\label{section:future_work}
\paragraph{Comparison with knockoffs}
In this paper, we proposed a variable selection procedure, the sequential CRT. In comparison with model-X knockoffs, the proposed sequential CRT is generally more powerful than model-X knockoffs as shown in the simulation studies in Section \ref{section:simulation}. Specifically, we observe a much more noticeable power gain of the sequential CRT than is observed by model-X knockoffs when the number of nonnulls is small. In addition to the power gain, we note that the sequential CRT has another advantage over model-X knockoffs: it is usually easier to sample from the conditional distribution of $X_j \mid X_{-j}$ than to generate knockoffs, especially for complicated joint distribution of $X$. 

\paragraph{Derandomizing the sequential CRT}
Like many other Model-X procedures (e.g.~Model-X knockoffs,
  distilled CRT, etc), the sequential CRT is a randomized
  procedure. In other words, different runs of the method produce
  might produce different selected sets. When the method is applied in
  practice, one would report those features whose selection frequency
  exceeds a threshold along with the corresponding
  frequencies. \citet{ren2020derandomizing} studies the problem of
  derandomizing knockoffs. It will be interesting to study whether it
  is possible to derandomize the proposed method so that results are
  more consistent across different runs.

\paragraph{Theoretically validating power gains}
While this paper demonstrates enhanced statistical power through
simulations, it would be interesting to theoretically validate power
gains. Intuitively, compared to Model-X knockoffs, the proposed method
effectively reduces the number of covariates by a factor of 2 when
computing feature importance statistics. It would be of interest to
understand theoretically how important such reduction is in terms of
statistical power.

\paragraph{Robustness to misspecification in the distribution of the covariates}
Another interesting direction for future work is to study the
    robustness of the sequential CRT to misspecification in the
    distribution of the covariates. When the distribution of $X$ is
    known only approximately, \citet{barber2020robust} quantifies the
    possible FDR inflation of model-X knockoffs; and
    \citet{berrett2020conditional} bounds the inflation in type-I
    error of the CRT. It will be interesting to evaluate the FDR
    inflation of the sequential CRT both empirically and
    theoretically.

\section*{Acknowledgements}

E. C. was supported by Office of Naval Research grant N00014-20-12157,
by the National Science Foundation grants OAC 1934578 and DMS 2032014,
and by the Simons Foundation under award 814641. S. L. was supported
by the National Science Foundation grant OAC 1934578.

\bibliographystyle{unsrtnat}
\bibliography{ref}

\newpage
\begin{appendix}
\section{Simulation Details and Additional Simulation Studies}
Software for our method is available from \url{https://github.com/lsn235711/sequential-CRT}, along with code to reproduce the analyses.
\label{section:simulation_details}
\subsection{Details of the simulation study in Section \ref{subsection:simulation_small}}
\label{subsection:simulation_details_small}
The samples are generated in the following way. In all examples, the samples are i.i.d.~copies. 
\begin{enumerate}
	\item 
		The explanatory variables $\bX$ are generated from an AR(1) model with correlation parameter $\rho = 0.5$.
		\begin{enumerate}
		\item Conditional linear model with $n = 300$
                  observations and $p = 300$ variables. We set
                  $Y = X^\top \beta + \epsilon$, where
                  $\epsilon \sim \mathcal{N}(0,1)$. The vector $\beta$
                  has $20$ non-zero entries equal to $A/\sqrt{n}$,
                  where the amplitude $A$ is a control parameter. The
                  non-zero entries of $\beta$ are chosen at random.
		
		\item Conditional logistic model with $n = 300$
                  observations and $p = 300$ variables.
                  $Y \mid X \sim \operatorname{Bern}\p{ 1/ \p{1 +
                      \exp\cb{-X^\top \beta}}}$. The vector $\beta$
                  has $20$ non-zero entries equal to $A/\sqrt{n}$,
                  where the amplitude $A$ is a control parameter. The
                  non-zero entries of $\beta$ are chosen at random.
		
		\item Conditional non-linear model with $n = 500$
                  observations and $p = 200$ variables. Conditional on
                  $X$,
                  $Y = \beta_0 \sum_{k = 1}^{10} \one \cb{X_{j_k} > 0} \one
                  \cb{X_{\ell_k} > 0} + \epsilon$, where
                  $\epsilon \sim \mathcal{N}(0,1)$ and
                  $\beta_0 = A/\sqrt{n}$. Again, $A$ is a control
                  parameter. The two
                  sets of indices in the regression function are each
                  of cardinality 10 and disjoint. They are chosen
                  uniformly at random.

		\item Conditional non-linear model with a binary
                  response, and $n = 500$ observations and $p = 200$
                  variables.
                  $Y \mid X \sim \operatorname{Bern}\p{ 1/ \p{1 +
                      \exp\cb{-\tilde{X}^\top \beta}}}$, where
                  $\tilde{X}_j = \one \cb{X_{j} > 0} - \one \cb{X_{j}
                    < 0}$. The vector $\beta$ has $20$ non-zero
                  entries equal to $A/\sqrt{n}$, where $A$ is a
                  control parameter. The non-zero entries of $\beta$
                  are chosen at random.
		\end{enumerate}

\item 
The explanatory variables $\bX$ are generated from an HMM model. 
The HMM model considered has 5 hidden states and 3 output states. The transition matrix is 
\[
		\begin{bmatrix}
		0.6 & 0.1 & 0.1 & 0.1 & 0.1\\
		0.1 & 0.6 & 0.1 & 0.1 & 0.1\\
		0.1 & 0.1 & 0.6 & 0.1 & 0.1\\
		0.1 & 0.1 & 0.1 & 0.6 & 0.1\\
		0.1 & 0.1 & 0.1 & 0.1 & 0.6\\
		\end{bmatrix}.
\] The emission probability matrix is
\[
	\begin{bmatrix}
		2/3 & 1/6 & 1/6 \\
		5/12 & 5/12 & 1/6\\
		1/6 & 2/3 & 1/6\\
		1/6 & 5/12 & 5/12\\
		1/6 & 1/6 & 2/3\\
	\end{bmatrix}.
\] The initial probabilities are $[0.2, 0.2, 0.2, 0.2, 0.2]$. 

\begin{enumerate}
		\item Conditional linear model with $n = 300$
                  observations and $p = 300$ variables. We set
                  $Y = X^\top \beta + \epsilon$, where
                  $\epsilon \sim \mathcal{N}(0,1)$. The vector $\beta$
                  has $20$ non-zero entries equal to $A/\sqrt{n}$,
                  where the amplitude $A$ is a control parameter. The
                  non-zero entries of $\beta$ are chosen at random.
		
		\item Conditional logistic model with $n = 300$
                  observations and $p = 300$ variables.
                  $Y \mid X \sim \operatorname{Bern}\p{ 1/ \p{1 +
                      \exp\cb{-(X - 2 \cdot \mathbf{1})^\top \beta}}}$. The vector $\beta$
                  has $20$ non-zero entries equal to $A/\sqrt{n}$,
                  where the amplitude $A$ is a control parameter. The
                  non-zero entries of $\beta$ are chosen at random.
		
		\item Conditional non-linear model with $n = 500$
                  observations and $p = 200$ variables. Conditional on
                  $X$,
                  $Y = \beta_0 \sum_{k = 1}^{10} \one \cb{X_{j_k} > 1.5} \one
                  \cb{X_{\ell_k} > 1.5} + \epsilon$, where
                  $\epsilon \sim \mathcal{N}(0,1)$ and
                  $\beta_0 = A/\sqrt{n}$. Again, $A$ is a control
                  parameter. The two
                  sets of indices in the regression function are each
                  of cardinality 10 and disjoint. They are chosen
                  uniformly at random.

		\item Conditional non-linear model with a binary
                  response, and $n = 500$ observations and $p = 200$
                  variables.
                  $Y \mid X \sim \operatorname{Bern}\p{ 1/ \p{1 +
                      \exp\cb{-\tilde{X}^\top \beta}}}$, where
                  $\tilde{X}_j = \one \cb{X_{j} > 1.5} - 2/3$. The vector $\beta$ has $20$ non-zero
                  entries equal to $A/\sqrt{n}$, where $A$ is a
                  control parameter. The non-zero entries of $\beta$
                  are chosen at random.
		\end{enumerate}

\end{enumerate}
We set the FDR threshold to be $q = 0.1$. We take $c = 0.1$ in Selective SeqStep+. The number of randomizations $B$ in Procedure \ref{alg:exact} and \ref{alg:inexact} are set to be $9$. To compute feature importance, the details of the blackbox algorithm we used are as follows: for lasso/glmnet, the regularization parameter is chosen using cross validation; for gradient boosting, we use the R-package XGBoost. We set the parameters \texttt{eta=0.05, max\_depth=2, nrounds = 100}. 

\subsection{Details of the simulation study in Section \ref{subsection:simulation_large}}
\label{subsection:simulation_details_large}
The samples are generated in the same way as in Section
\ref{subsection:simulation_details_small}, however, with a larger
number of observations $n$, a larger number of variables $p$, and a
larger number of nonnulls $k$. The methods are also the same as in
\ref{subsection:simulation_details_small}. With the same labeling as before, we set the parameters as follows. 
\begin{enumerate}
	\item For the AR(1) model:
		\begin{enumerate}
		\item $n = 1000$, $p = 1000$, and $k = 50$;
		
		\item $n = 1000$, $p = 1000$, and $k = 50$; 
		
		\item $n = 1200$, $p = 500$, and $k = 50$; 
		
		\item $n = 1200$, $p = 500$, and $k = 50$.
		\end{enumerate}
\item 
For the HMM model:

		\begin{enumerate}
		\item $n = 1000$, $p = 1000$, and $k = 50$. 
		
		\item $n = 1000$, $p = 1000$, and $k = 50$. 
		
		\item $n = 1200$, $p = 500$, and $k = 50$. 
		
		\item $n = 1200$, $p = 500$, and $k = 50$.
		\end{enumerate}

\end{enumerate}

\subsection{Details of the simulation study in Section \ref{subsection:simulation_num_nonnull}}
\label{subsection:simulation_details_num_nonnull}
The samples are generated in the same way as in Section
\ref{subsection:simulation_details_small}, however, with a fixed
control parameter $A$, and a varying number of nonnulls
$k \in \cb{5, 10, \dots, 50}$. The methods are also the same as in
\ref{subsection:simulation_details_small} and we only consider the
AR(1) model. With the same labeling as before, we set the parameters
as follows: 
		\begin{enumerate*}[(a),series = tobecont] 
\item $A = 1.5$;	\item $A = 8$; 		\item $A = 15$; 		\item $A = 20$.
		\end{enumerate*}

\subsection{Details of the simulation study in Section \ref{subsection:simulation_Bc}}
\label{subsection:simulation_details_Bc}
The FDR threshold is set to be $q = 0.1$. To compute feature importance statistic, we use lasso for the linear case and glmnet for the logistic case, where the regularization parameter is chosen using cross validation. 
For the ``small dataset" setting, the samples are generated in the same way as in Section \ref{subsection:simulation_details_small}. We only consider the AR(1) model. We fix the control parameter $A$ as follows:
\begin{center}
\begin{tabular}{ |c|c|c|c| } 
 \hline
linear, low SNR & linear, high SNR & logistic, low SNR  & logistic high SNR \\ 
 \hline
$A = 3.5$ & $A = 5$ & $A = 12$ & $A = 30$ \\ 
 \hline
\end{tabular}
\end{center}
For the ``large dataset" setting, the samples are generated in the same way as in Section \ref{subsection:simulation_details_large}. We only consider the AR(1) model. We fix the control parameter $A$ as follows:
\begin{center}
\begin{tabular}{ |c|c|c|c| } 
 \hline
linear, low SNR & linear, high SNR & logistic, low SNR  & logistic high SNR \\ 
 \hline
$A = 3.5$ & $A = 5$ & $A = 10$ & $A = 20$ \\ 
 \hline
\end{tabular}
\end{center}

Figure \ref{fig:Bc_FDR} compares the performance of the sequential CRT (one-shot version) in terms of the FDR with various different values of $c$ and $B$. It appears that the FDR is well controlled for all combinations of $c$ and $B$. 

\begin{figure}
\centering
\begin{subfigure}[b]{0.9\textwidth}
      \caption{Small datasets}
         \centering
         \includegraphics[width = \textwidth]{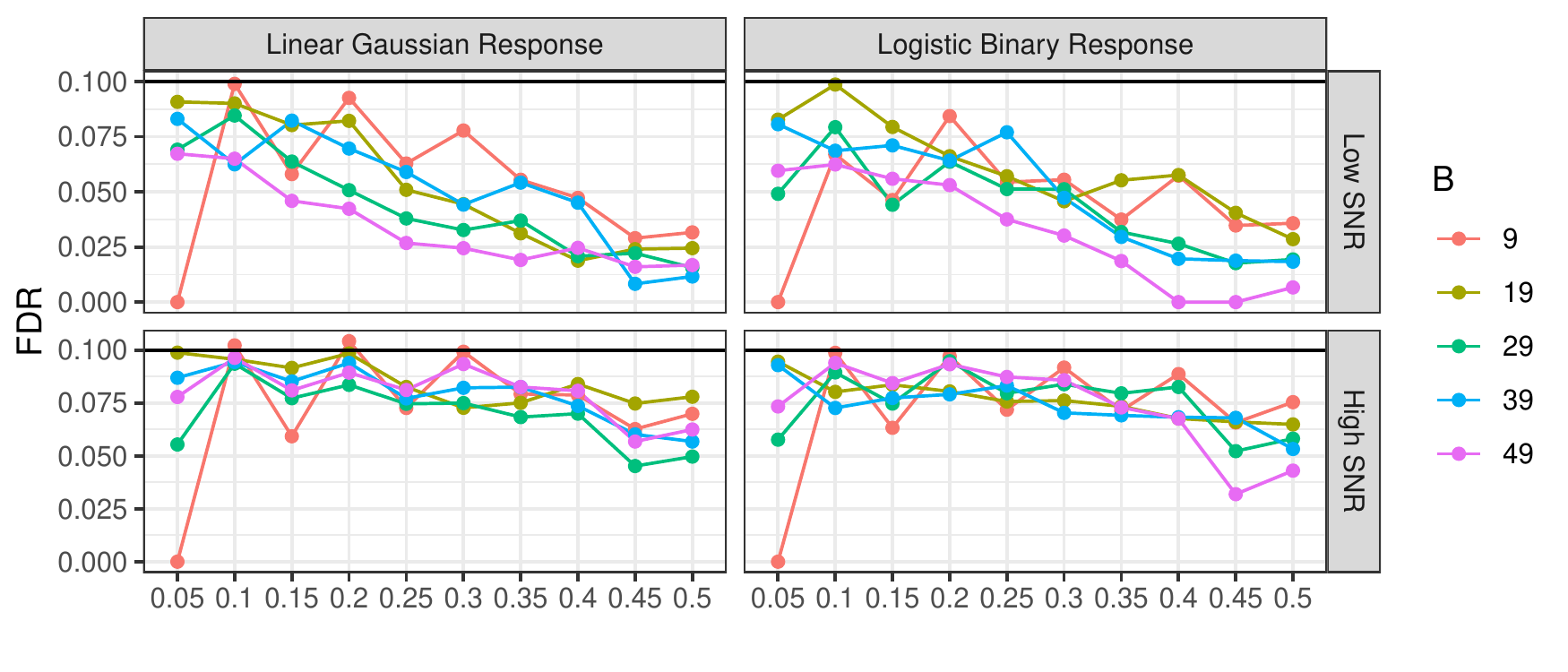}
         \label{fig:Bc_small_FDR}
 \end{subfigure}
\begin{subfigure}[b]{0.9\textwidth}
\noindent\rule{\textwidth}{0.4pt}
  \caption{Large datasets}
\centering
 \includegraphics[width = \textwidth]{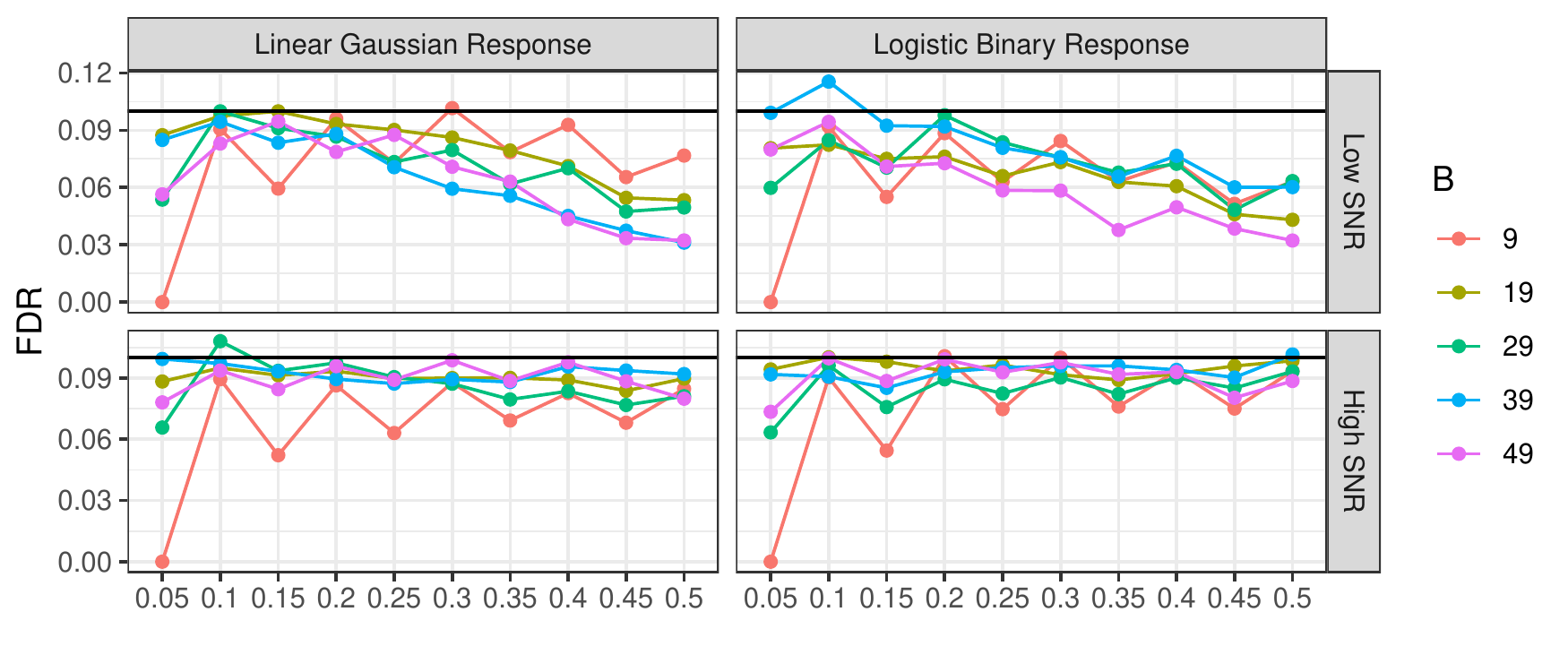}
 \label{fig:Bc_large_FDR}
\end{subfigure}
\caption{FDR of the sequential CRT (one-shot) with different choices of threshold $c$ and number of randomizations $B$. The nominal false discovery
  rate level is 10\%. Results are averaged over 100 independent
  experiments.}
\label{fig:Bc_FDR}
\end{figure}

\subsection{Implementation details of Table \ref{table:time}}
\label{subsection:simulation_details_table}
The explanatory variables $\bX$ are generated from an AR(1) model with
correlation parameter $\rho = 0.5$. Other details are as in Section
\ref{subsection:simulation_details_small}.

\subsection{Implementation details of Figure \ref{fig:hist_max_aj}}
\label{subsection:simulation_details_cond_prob}
We simulate $n = 200$ observations and $p = 120$ variables.  The
explanatory variables $\bX$ are jointly Gaussian with mean $0$ and
variance $\Sigma$. $\Sigma$ is a block diagonal matrix with block size
$3$. The non-zero off-diagonal entries of $\Sigma$ are equal to 0.3. Conditional on $X$, $Y = X^\top \beta + \epsilon$, where $\epsilon \sim \mathcal{N}(0,1)$. The vector $\beta$ has $30$ non-zero entries equal to
$3/\sqrt{n}$. The non-zero entries of $\beta$ are chosen at
random. The $\pvals$ are obtained using Algorithm \ref{alg:crt} where
the test statistics are computed using the lasso. The number of
randomizations $B$ is set to $19$. We set $q = 0.1$ and $c =
0.3$. Each $a_j$ is estimated as an average of 5000 binary
variables. The histogram (Figure \ref{fig:hist_max_aj}) is based on
500 samples.

\subsection{Comparison of the Sequential CRT with other benchmark methods}
\label{section:comparison_benchmark}
 We compare the proposed split version (Procedure \ref{alg:exact}) and symmetric statistic version (Procedure \ref{alg:inexact}) of the sequential CRT with the following methods:
\begin{enumerate}
\item \emph{Knockoffs} \citep{candes2018panning}. Knockoffs are constructed with the Gaussian semi-definite optimization algorithm. We take the
feature importance statistic to be the glmnet coefficient difference. 
\item \emph{Distilled CRT} \citep{liu2020fast}. We consider both $\operatorname{d}_0$CRT and $\operatorname{d}_{\operatorname{I}}$CRT; we refer to Section 2.3 and 2.4 of \citep{liu2020fast} for specific constructions of the dCRT. In the continuous response case, the distillation step is done by a cross-validated lasso; in the binary response case, the distillation step is done by $L_1$-penalized logistic regression. 
\item \emph{HRT} \citep{tansey2018holdout}. We implement Algorithm 1 of \citep{tansey2018holdout} with a data split of 50\%-50\%. In the continuous response case, the algorithm is implemented with a cross-validated lasso; in the binary response case, it is implemented with a cross-validated $L_1$-penalized logistic regression. 
\item \emph{Gaussian Mirror} \citep{xing2019controlling}. We implement the Gaussian Mirror with the \texttt{gm()} function in the GM package (\url{https://github.com/BioAlgs/GM}). 

\end{enumerate}

We run the sequential CRT with one-shot CRT. 
We consider similar settings as in Section \ref{subsection:simulation_large}, and focus on the settings of AR model with a conditional linear model and AR model with a conditional logistic model. 
Other details can be found in Section \ref{subsection:simulation_large} and Appendix \ref{subsection:simulation_details_large}. Figure \ref{fig:large_comparison_ar} compares the performance of the above methods in terms of empirical false discovery rate and power averaged over 100 independent replications. The sequential CRT (symmetric statistics) appears to have the highest power among all methods. In terms of the false discovery rate, all methods except the gaussian mirror control the false discovery rate at the desired level. The gaussian mirror appears to have a huge FDR inflation when the signal to noise ratio is small. Hence, from a practical point of view, it may not the most
  appealing method when we have no prior information on the signal
  strength.

\begin{figure}
\includegraphics[width = \textwidth]{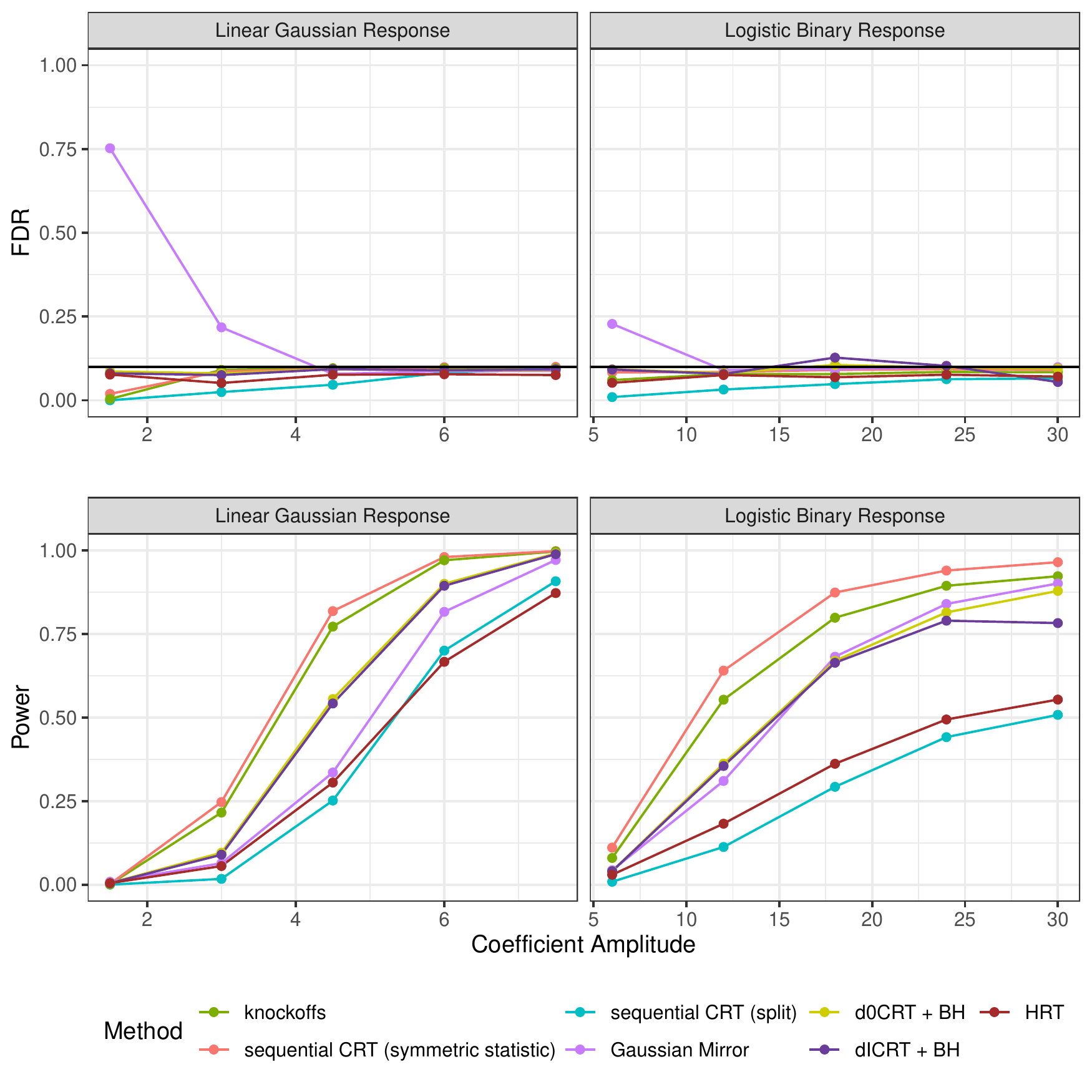}
\caption{Performance of the proposed split version and symmetric statistic version of the sequential CRT compared to other benchmarks on synthetic datasets. The nominal false discovery rate level is 10\%. Results are averaged over 100 independent experiments.}
\label{fig:large_comparison_ar}
\end{figure}

\section{Proofs}
\subsection{Proof of Theorem \ref{theo:exchangeable}}
\label{subsection:proof_exch}
For the sake of notation, we write $m = p$, i.e. we let $m$ be the number of hypotheses/variables, to distinguish from the ``p" in $\pvals$: $p_1, p_2, $ etc. 

\subsubsection{Proof of upper bound \eqref{eqn:exch_bound1}}
\label{subsection:proof_exch1}
To bound the FDR, we modify arguments from the proof of Lemma 1 and
Theorem 3 in Supplement of \citep{barber2015controlling}. We include a
part of the statement in \citep[Lemma 1]{barber2015controlling} for reference:

\begin{quote}
``For $k=m, m-1, \ldots, 1,0,$ put $V^{+}(k)=$ $\#\left\{ j \in \Ho : 1 \leq j \leq k, p_{j} \leq c\right\}$ and $V^{-}(k)=\# \{  j \in \Ho : 1 \leq j \leq k, p_{j}>c \}$ with the convention that $V^{\pm}(0)=0 .$ Let $\mathcal{F}_{k}$ be the filtration defined by knowing all the nonnull $\pvals$, as well as $V^{\pm}\left(k^{\prime}\right)$ for all $k^{\prime} \geq k .$ Then the process
$$
M(k)=\frac{V^{+}(k)}{1+V^{-}(k)}
$$
is a super-martingale running backward in time with respect to $\mathcal{F}_{k}$."
 \end{quote}

The conclusion above assumes that the null $p$ -values are i.i.d.,
satisfy $p_{j} \geq$ Unif[0,1], and are independent from the nonnulls.
Here, we are no longer in the i.i.d.~setting of \citep[Lemma
1]{barber2015controlling}. Yet, the same proof goes through with
exchangeability. This mean that the above result still holds under the
conditions that the null $\pvals$ are exchangeable given nonnull
$\pvals$ and that marginally, null $\pvals$ satisfy $p_{j} \geq$
Unif[0,1].

As in \citep{barber2015controlling}, the $\hat{k}$ defined in \eqref{eqn:seqstep}
is a stopping time with respect to the backward filtration
$\left\{\mathcal{F}_{k}\right\}$ since
$\{\hat{k} \geq k\} \in \mathcal{F}_{k}$. By the optional stopping
time theorem for super-martingales,
\[
\EE{ M(\hat{k}) \big| \mathcal{F}_m}
 \leq M(m)
=\frac{\#\left\{ j \in \Ho: p_{j} \leq c\right\}}{1+\#\left\{j \in \Ho: p_{j}>c\right\}}. 
\]
As in \citep{barber2015controlling}, we write $V=\#\cb{j \in \Ho,  j \leq \hat{k}: p_{j} \leq c}$ and $R=\#\left\{j \leq \hat{k}: p_{j} \leq c\right\}$. 
\begin{align*}
\EE{\FDP | \mathcal{F}_m} &= \mathbb{E}\left[\frac{V}{R \vee 1} \Big| \mathcal{F}_m \right]
=\mathbb{E}\left[\frac{V}{R \vee 1} \cdot \one\cb{\hat{k}>0} \Big| \mathcal{F}_m \right]\\
&=\mathbb{E}\left[\frac{\#\left\{j \in \Ho, j  \leq \hat{k}: p_{j} \leq c\right\}}{1+\#\left\{j \in \Ho, j  \leq \hat{k}: p_{j}>c\right\}} \cdot\left(\frac{1+\#\left\{j \in \Ho, j  \leq \hat{k}: p_{j}>c\right\}}{\#\left\{j \leq \hat{k}: p_{j} \leq c\right\} \vee 1} \cdot \one\cb{\hat{k}>0}\right) \Bigg| \mathcal{F}_m \right] \\
&\leq \p{ \mathbb{E}\left[\frac{\#\left\{j \in \Ho, j  \leq \hat{k}: p_{j} \leq c\right\}}{1+\#\left\{j \in \Ho, j  \leq \hat{k}: p_{j}>c\right\}} \Bigg| \mathcal{F}_m \right] \cdot \frac{1-c}{c} \cdot q} \wedge 1\\
& \leq \p{\frac{\#\left\{ j \in \Ho: p_{j} \leq c\right\}}{1+\#\left\{j \in \Ho: p_{j}>c\right\}}  \cdot \frac{1-c}{c} \cdot q }\wedge 1\\
&\leq \p{\frac{\#\left\{ j \in \Ho: p_{j} \leq c\right\}}{1+\#\left\{j \in \Ho: p_{j}>c\right\}} \wedge \frac{c}{(1-c)q} } \frac{1-c}{c}q.
\end{align*}
This implies that
\begin{align*}
\FDR
&= \EE{\EE{\FDP | \mathcal{F}_m}}
\leq \frac{1-c}{c}q \EE{\p{\frac{\#\left\{ j \in \Ho: p_{j} \leq c\right\}}{1+\#\left\{j \in \Ho: p_{j}>c\right\}} \wedge \frac{c}{(1-c)q} }}\\
&=\frac{(1-c)q}{c}  \sum_{k = 0}^m \pi_k \p{ \frac{k}{1+m-k}  \wedge \frac{c}{(1-c)q}},
\end{align*}
where $\pi_k = \PP{\#\cb{i \in \Ho: p_i \leq c} = k}$. 

Thus, bounding the FDR reduces to an optimization problem:
\begin{alignat}{1}
\label{eqn:opt_no_var}
\text{maximize}_{\pi_k}  \quad & \sum_{k = 0}^m \pi_k \p{ \frac{k}{1+m-k}  \wedge \frac{c}{(1-c)q}}\\
\text{subject to} \quad & \sum \pi_k = 1, \pi_k \geq 0, \sum k \pi_k \leq c m. \nonumber
\end{alignat}
More rigorously, let $\Vopt$ be the optimal value of
\eqref{eqn:opt_no_var}. Then our previous analysis implies that
$\FDR \leq \frac{(1-c) q}{c}\Vopt$. Note that with
$\pi_k = \PP{\#\cb{i \in \Ho: p_i \leq c} = k}$, we have
$\sum k \pi_k \leq c$ since 
$\sum k \pi_k = \EE{\#\cb{i \in \Ho: p_i \leq c}} = \sum_{i \in \Ho}
\PP{p_i \leq c} \leq c m$.

It remains to solve \eqref{eqn:opt_no_var}. The optimal value of \eqref{eqn:opt_no_var} is indeed the same as that of \eqref{eqn:opt_no_var2}:
\begin{alignat}{1}
\label{eqn:opt_no_var2}
\text{maximize}_{\pi_k}  \quad & \sum_{k = 0}^{m'} \pi_k \p{ \frac{k}{1+m-k}} \\
\text{subject to} \quad & \sum_{k = 0}^{m'} \pi_k = 1, \pi_k \geq 0, \sum_{k = 0}^{m'} k \pi_k \leq c m, \nonumber
\end{alignat}
where $m' = \lfloor \frac{c m}{q + c(1-q)} \rfloor$. This is because
for $k > m'$, $\frac{k}{1 + m - k} > \frac{c}{(1-c)q}$. Hence for any
optimal solution of \eqref{eqn:opt_no_var} such that $\pi_k > 0$ for
some $k > m'$, we can move the mass to $m'$, i.e., define
$\pi'_{m'} = \sum_{k \geq m'} \pi_k$. Then the new $\pi'$ still
satisfies the constraints in \eqref{eqn:opt_no_var} and achieves the same objective
value.

To solve \eqref{eqn:opt_no_var2}, note that since $\frac{k}{1 + m - k}$
is convex in $k$, and the constraints are linear in $k$, the optimal value is achieved when the probability $\pi$ has its mass on the boundary. Formally, 
the optimal value is achieved when $\pi_k = 0$ for
all $k$ except for $0$ and $m'$. To see this, assume that there is a $k_0 \neq 0, m'$ such that $\pi_{k_0} \neq 0$. We will see that if we move the mass at $k_0$ to $0$ and $m'$, then the constraints are still satisfied but the objective function will have a larger value. More precisely, if we define $\tilde{\pi}$ by taking $\tilde{\pi}_0 = \pi_0 + \pi_{k_0}(m' - k_0)/m'$, $\tilde{\pi}_{m'} = \pi_{m'} + \pi_{k_0} k_0/m'$, and $\tilde{\pi}_{k_0} = 0$, then $\sum_{k = 0}^{m'} \tilde{\pi}_k = \sum_{k = 0}^{m'} \pi_k = 1$, and $\sum_{k = 0}^{m'} k \tilde{\pi}_k = \sum_{k = 0}^{m'} k \pi_k \leq c m$. The objective function $\sum_{k = 0}^{m'} \tilde{\pi}_k \p{ \frac{k}{1+m-k}} > \sum_{k = 0}^{m'} \pi_k \p{ \frac{k}{1+m-k}}$, because $\frac{k}{1 + m - k}$
is strictly convex in $k$. 

Therefore the optimal value of \eqref{eqn:opt_no_var2} is achieved when $\pi_k = 0$ for all $k$ except for $0$ and $m'$. 
Thus the optimal value
$\Vopt = \frac{\pi_{m'} m'}{1 + m - m'}$, where $m' \pi_{m'} = c
m$. As
$m' = \lfloor \frac{c m}{q + c(1-q)} \rfloor \leq \frac{c m}{q +
  c(1-q)}$, we have
\[ \Vopt = \frac{\pi_{m'} m'}{1 + m - m'} \leq \frac{c m}{1 + m -  \frac{c m}{q + c(1-q)}} \leq \frac{c(q + c(1-q))}{q(1-c)}. \]
Hence
\[\FDR \leq \frac{(1-c)q}{c} \frac{c(q + c(1-q))}{q(1-c)} = q + c(1-q).\]

\subsubsection{Proof of upper bound \eqref{eqn:exch_bound2}}
\label{subsection:proof_exch2}
We work here with the additional constraint that for pairs of nulls
$(i,j)$, $\Corr{\one \cb{p_i\leq c}, \cb{p_j\leq c}} \leq \rho$. Put
$c_0 = \PP{p_i\leq c}$ and $m_0 = \#\cb{i \in \Ho}$. Then
$\Cov{\one \cb{p_i\leq c}, \cb{p_j\leq c}} \leq \rho c_0 (1-c_0)$. Also,
\begin{align*} 
\Var{\#\cb{i \in \Ho: p_i \leq c} } & = m_0 \Var{ \one \cb{p_i\leq c}}
+ m_0(m_0 - 1) \operatorname{Cov} [ \one \cb{p_i\leq c},
\one \cb{p_j\leq c}] \\ & \leq m_0(1 + \rho(m_0 - 1))c_0(1-c_0).
\end{align*}
Alternatively, we can write the variance as a function of
the $\pi_k$'s:
$$\Var{\#\cb{i \in \Ho: p_i \leq c} } = \sum_{k = 0}^m (k - c_0 m_0)^2
\pi_k.$$ Thus the constraint that
$\Corr{\one \cb{p_i\leq c}, \cb{p_j\leq c}} \leq \rho$ can be
rewritten as
$$\sum_{k = 0}^m (k - c_0 m_0)^2 \pi_k \leq m_0(1 + \rho(m_0 -
1))c_0(1-c_0).$$

Following the analysis in Section \ref{subsection:proof_exch1}, bounding the FDR reduces to solving an
optimization problem with decision variables $\{\pi_k\}$, $m_0$ and $c_0$: 
\begin{alignat}{2}
\label{eqn:opt1}
&\text{maximize} \quad && \sum_{k = 0}^m \pi_k \p{ \frac{k}{1+m-k}  \wedge \frac{c}{(1-c)q}}\\
&\text{subject to} \quad && \sum_{k = 0}^m \pi_k = 1 , \pi_k \geq 0,  \,\, \sum_{k = 0}^m \pi_k = c_0 m_0, \,\, c_0 \leq c, \,\,  m_0 \leq m,\nonumber\\
&  && \sum_{k = 0}^m (k - c_0)^2 \pi_k \leq m_0(1 + \rho(m_0 - 1))c_0(1-c_0). \nonumber
\end{alignat}
Specifically, $\FDR \leq \frac{(1-c)q}{c} \Vopt\eqref{eqn:opt1}$,
where $\Vopt\eqref{eqn:opt1}$ is the optimal value of the program
\eqref{eqn:opt1}. The optimization problem
is equivalent (in the sense that
$\Vopt\eqref{eqn:opt1} = \Vopt\eqref{eqn:opt2}$) to:
\begin{alignat}{2}
\label{eqn:opt2}
&\text{maximize}  \quad && \sum_{k = 0}^{m'} \pi_k \p{\frac{k}{1+m-k} } \\
&\text{subject to} \quad && \sum_{k = 0}^{m'} \pi_k = 1 , \,\, \pi_k \geq 0,  \,\, \sum_{k = 0}^{m'} \pi_k = c_0 m_0, \,\, c_0 \leq c, m_0 \leq m,\nonumber\\
&  && \sum_{k = 0}^{m'} (k - c_0)^2 \pi_k \leq m_0(1 + \rho(m_0 - 1))c_0(1-c_0), \nonumber
\end{alignat}
where $m' = \lfloor  \frac{c m}{q + c(1-q)} \rfloor$. 
For now, fix $\sigma > 0$ and $c_1 \leq c$ and focus on the case where $\sum_{k = 0}^{m'} (k - c_0)^2 \pi_k = m^2 \sigma^2$ and $\sum_{k = 0}^m \pi_k = c_1 m = c_0 m_0$. We will optimize over the choice of $\sigma$ and $c_1$ later. That is, consider
\begin{alignat}{2}
\label{eqn:opt3}
&\text{maximize} \quad && \sum_{k = 0}^{m'} \pi_k \p{\frac{k}{1+m-k} }\\
&\text{subject to} \quad && \sum_{k = 0}^{m'} \pi_k = 1 , \pi_k \geq 0,  \sum_{k = 0}^{m'} \pi_k = c_1 m,\nonumber\\
&  && \sum_{k = 0}^{m'} (k - c_1)^2 \pi_k = \sigma^2 m^2, \nonumber
\end{alignat}
with decision variables $\{\pi_k\}$. $\Vopt\eqref{eqn:opt3}$ is a function of $\sigma$, $c_0$ and $m_0$ and we have: 
\begin{equation}
\label{eqn:opt_value_inequality}
 \Vopt\eqref{eqn:opt2} = \max\cb{\Vopt\eqref{eqn:opt3} : \sigma^2 \leq m_0(1 + \rho(m_0 - 1))c_0(1-c_0)/m^2, c_0 \leq c, m_0 \leq m}.
\end{equation}
As for the optimization problem \eqref{eqn:opt3}, this is an LP with
three equality constraints. By Carath\'eodory's theorem, any optimal
solution must have at most three non-zero $\pi_k$'s. Hence, the problem
becomes ($\Vopt\eqref{eqn:opt3} = \Vopt\eqref{eqn:opt4}$)
\begin{alignat}{2}
\label{eqn:opt4}
&\text{maximize} \quad && \pi_1\frac{k_1}{ 1 + m - k_1} +\pi_2\frac{k_2}{ 1 + m - k_2} + \pi_3\frac{k_3}{ 1 + m - k_3} \\
&\text{subject to} \quad && \pi_1 + \pi_2 + \pi_3 = 1, \, (\pi_1, \pi_2, \pi_3) \ge 0,  \nonumber\\
&&&\pi_1 k_1 + \pi_2 k_2 + \pi_3 k_3 = c_1 m, \nonumber\\
&  && \pi_1 (k_1 - c_1)^2 + \pi_2 (k_2 - c_1)^2 + \pi_3 (k_3 - c_1)^2  = \sigma^2 m^2,\nonumber\\
&  && (k_1, k_2, k_3) \leq  \lfloor  \frac{c m}{q + c(1-q)} \rfloor, \nonumber
\end{alignat}
where the last inequality says that all components of the left-hand
side are upper bounded by the right-hand side.  Set $x_i = k_i/m$,
$1 \le i \le 3$, and relax the constraints by allowing  $x_1$ to
take any real value in the interval. We also increase the objective
for simplification and consider
\begin{alignat}{2}
\label{eqn:opt5}
&\text{maximize} \quad &&\pi_1 \frac{x_1}{1-  x_1} + \pi_2 \frac{x_2}{1 - x_2} + \pi_3 \frac{x_3}{1 - x_3} \quad \\
& \text{subject to} \quad &&\pi_1 + \pi_2 + \pi_3 = 1, \nonumber\\
& && \pi_1 x_1 + \pi_2 x_2 + \pi_3 x_3 = c_1, \nonumber\\
 & && \pi_1 (x_1-c_1)^2 + \pi_2 (x_2-c_1)^2 + \pi_3 (x_3-c_1)^2 = \sigma^2, \nonumber\\
 & &&  (\pi_1, \pi_2, \pi_3) \geq 0, \,\, 0 \leq  (x_1, x_2, x_3) \leq  \frac{c}{q + c(1-q)},  \nonumber
\end{alignat}
so that $\Vopt\eqref{eqn:opt4} \leq  \Vopt\eqref{eqn:opt5}$. By lemma \ref{lemm:opt},  
\begin{align*}
 \Vopt\eqref{eqn:opt5} \leq \frac{\alpha}{1 - \alpha} - \frac{(\alpha - c_1)^2}{(1-\alpha)(\sigma^2 + (1-c_1)(\alpha - c_1))},
\end{align*}
where $\alpha = \frac{c}{q + c(1-q)}$. The upper bound is clearly an increasing function of $\sigma^2$ and $c_1$. Hence, for any $\sigma^2 \leq m_0(1 + \rho(m_0 - 1))c_0(1-c_0)/m^2 \leq (1 + \rho(m - 1))/m c_0(1-c_0)$ and any $c_1 = c_0 m_0/m \leq c_0$, 
\begin{align*}
 \Vopt\eqref{eqn:opt5} &\leq \frac{\alpha}{1 - \alpha} - \frac{(\alpha - c_1)^2}{(1-\alpha)(\sigma^2 + (1-c_1)(\alpha - c_1))}\\
& \leq \frac{\alpha}{1- \alpha} - \frac{(\alpha - c_0)^2}{(1-\alpha)(\tilde{\rho}c_0(1-c_0) + (1-c_0)(\alpha - c_0))},
\end{align*}
where $\tilde{\rho} = ((m-1)\rho + 1)/m$. 
As
\[ \frac{d}{d c_0}  \frac{(\alpha - c_0)^2}{(\tilde{\rho}c_0(1-c_0) + (1-c_0)(\alpha - c_0))} = -\frac{ (\alpha-c_0)(1-\alpha) + \tilde{\rho}(\alpha + c_0 - 2 \alpha c_0))}{(c_0-1)^{2}(\alpha+c_0 \rho-c_0)^{2}}<0, \]
the bound $\frac{\alpha}{1- \alpha} - \frac{(\alpha - c_0)^2}{(1-\alpha)(\tilde{\rho}c_0(1-c_0) + (1-c_0)(\alpha - c_0))}$ is an increasing function of $c_0$. Hence for any $c_0 \leq c$,
\begin{align*}
 \Vopt\eqref{eqn:opt5} & \leq \frac{\alpha}{1- \alpha} - \frac{(\alpha - c_0)^2}{(1-\alpha)(\tilde{\rho}c_0(1-c_0) + (1-c_0)(\alpha - c_0))}\\
& \leq \frac{\alpha}{1- \alpha} - \frac{(\alpha - c)^2}{(1-\alpha)(\tilde{\rho}c(1-c) + (1-c)(\alpha - c))}. 
\end{align*}
Together with \eqref{eqn:opt_value_inequality}, this implies that
\begin{align*}
 \Vopt\eqref{eqn:opt1}  
& \leq \frac{\alpha}{1- \alpha} - \frac{(\alpha - c)^2}{(1-\alpha)(\tilde{\rho}c(1-c) + (1-c)(\alpha - c))},
\end{align*}
where $\alpha = \frac{c}{q + c(1-q)}$, $\tilde{\rho} = ((m-1)\rho + 1)/m$. 
Thus 
\begin{align*}
 \FDR &\leq \frac{(1-c)q}{c}\Vopt\eqref{eqn:opt1}  
\leq  \frac{(1-c)q}{c}\p{ \frac{\alpha}{1- \alpha} - \frac{(\alpha - c)^2}{(1-\alpha)(\tilde{\rho}c(1-c) + (1-c)(\alpha - c))} }\\
& = q + \frac{\tilde{\delta}}{1 + \beta \tilde{\delta}} \sqb{\frac{c}{1-c} - \frac{c - c\tilde{\delta}}{1 - (c-c\tilde{\delta})}q},
\end{align*}
where $\tilde{\rho} = \frac{(m-1)\rho + 1}{m}$, $\beta= \frac{c + (1-c)q}{(1-c)(1-q)}$, and $\tilde{\delta} = \tilde{\rho} \frac{c(1-q) + q}{c(1-q)}$. Upon setting $\delta = \rho  \frac{c(1-q) + q}{c(1-q)}$, $\tilde{\delta} \leq \delta$, whence,  
\[ \FDR \leq  q + \frac{\tilde{\delta}}{1 + \beta \tilde{\delta}} \sqb{\frac{c}{1-c} - \frac{c - c\tilde{\delta}}{1 - (c-c\tilde{\delta})}q} \leq q + \frac{\delta}{1 + \beta \delta} \sqb{\frac{c}{1-c} - \frac{c - c\delta}{1 - (c-c\delta)}q}. \]

\begin{lemm}
\label{lemm:opt}
Fix $1 > \alpha > c > 0$.  The problem
\begin{alignat}{2}
\label{eqn:opt_3_terms}
&\text{\em maximize} \quad &&\pi_1 \frac{x_1}{1-  x_1} + \pi_2 \frac{x_2}{1 - x_2} + \pi_3 \frac{x_3}{1 - x_3} \quad \\
& \text{\em subject to} \quad &&\pi_1 + \pi_2 + \pi_3 = 1, \nonumber\\
& && \pi_1 x_1 + \pi_2 x_2 + \pi_3 x_3 = c, \nonumber\\
 & && \pi_1 (x_1-c)^2 + \pi_2 (x_2-c)^2 + \pi_3 (x_3-c)^2 = \sigma^2, \nonumber\\
 & &&  (\pi_1, \pi_2, \pi_3) \geq 0, \,\, (x_1, x_2, x_3) \leq  \alpha, \nonumber
\end{alignat}
has optimal value 
\[ \Vopt = \frac{\alpha}{1 - \alpha} - \frac{(\alpha - c)^2}{(1-\alpha)(\sigma^2 + (1-c)(\alpha - c))}. \]

\end{lemm}

\begin{proof}
  Let $\pi^{\star}$ and $x^{\star}$ be
  an optimal solution with objective value $\Vopt$. If $\pi_i^{\star} > 0$
  and $x_i^{\star} < \alpha$ for all $i = 1,2,3$, then by complementary slackness, there exist
  $\lambda$ and $\gamma$ such that
\[ \frac{\pi^{\star}_i}{(1-x_i^{\star})^2} + \lambda \pi_i^{\star} + \gamma \pi_i^{\star} (2x_i^{\star} - c) = 0 \]
holds for $i = 1,2,3$. This further implies that
\[ \frac{1}{(1-x_i^{\star})^2} + \lambda + \gamma (2x_i^{\star} - c) = 0. \] 
The left-hand side is strictly convex in $x$, thus this equation has at most two distinct solutions for $x$. Thus any optimal
solution of \eqref{eqn:opt_3_terms} either has some $x_i$'s taking on
the value $\alpha$ or two $x_i$'s are the same. We analyze each case
below.

\paragraph{Optimal solution has some entries taking on the value $\alpha$.}
Without loss of generality, assume that the optimal solution has
$x_3 = \alpha$. Then
\eqref{eqn:opt_3_terms} is equivalent to 
\begin{alignat*}{2}
&\text{maximize} \quad &&\pi_1 \frac{x_1}{1-  x_1} + \pi_2 \frac{x_2}{1 - x_2} + \pi_3 \frac{\alpha}{1 - \alpha} \quad \\
& \text{subject to} \quad &&\pi_1 + \pi_2 + \pi_3 = 1,\\
& && \pi_1 x_1 + \pi_2 x_2 + \pi_3 \alpha = c, \nonumber\\
 & && \pi_1 (x_1-c)^2 + \pi_2 (x_2-c)^2 + \pi_3 (\alpha-c)^2 = \sigma^2, \\
 & &&  (\pi_1, \pi_2, \pi_3) \geq 0, \,\, (x_1, x_2) \leq  \alpha. \nonumber
\end{alignat*}
Write $\pi_3 = 1 - \pi_1 - \pi_2$. The above program is equivalent to
\begin{alignat*}{2}
&\text{maximize} \quad &&- \frac{\pi_1 (\alpha - x_1)}{(1-  x_1)(1-\alpha)} - \frac{ \pi_2(\alpha - x_2)}{(1-  x_1)(1-\alpha)} +  \frac{\alpha}{1 - \alpha} \quad \\
& \text{subject to} \quad &&\pi_1 (\alpha - x_1) + \pi_2 (\alpha - x_2 )  = \alpha - c, \\
 & && \pi_1 (\alpha - x_1)(\alpha + x_1 - 2c) + \pi_2 (\alpha - x_2)(\alpha + x_2 - 2c)  = (\alpha - c)^2- \sigma^2, \\
 & &&  (\pi_1, \pi_2, \pi_3) \geq 0, \,\, (x_1, x_2, x_3) \leq  \alpha. \nonumber
\end{alignat*}
Put $\theta_1 = \pi_1 (\alpha - x_1)$ and $\theta_2 = \pi_1 (\alpha - x_1)$. The program above can be simplified to  
\begin{alignat}{2}
\label{eqn:opt_theta_simp}
&\text{maximize} \quad &&- \frac{1-\theta_1}{(1-  x_1)(1-\alpha)} - \frac{1-\theta_2}{(1-  x_1)(1-\alpha)} +  \frac{\alpha}{1 - \alpha} \quad \\
& \text{subject to} \quad &&\theta_1 + \theta_2  = \alpha - c, \nonumber\\
 & && \theta_1(\alpha + x_1 - 2c) + \theta_2(\alpha + x_2 - 2c)  = (\alpha - c)^2- \sigma^2, \nonumber\\
 & &&  (\theta_1, \theta_2, \theta_3) \geq 0, \,\, (x_1, x_2, x_3) \leq  \alpha. \nonumber
\end{alignat}
Now the constraints in \eqref{eqn:opt_theta_simp} are linear in $x$
and the objective is convex. Hence, the optimal solution either has
$x_1^{\star} = \alpha$ or $x_2^{\star} = \alpha$. In other words, the optimal solution either has $x_1^{\star} = x_3^{\star} = \alpha$ or $x_2^{\star} = x_3^{\star} = \alpha$. Thus the set $\cb{x_1^{\star}, x_2^{\star}, x_3^{\star}}$ has at most two distinct elements, with one of them being $\alpha$.

\paragraph{Optimal solution has two entries taking on the same value.}
Without loss of generality, assume $\pi_3 = 0$. Then problem \eqref{eqn:opt_3_terms} becomes
\begin{alignat*}{2}
&\text{maximize} \quad &&\pi_1 \frac{x_1}{1-  x_1} + \pi_2 \frac{x_2}{1 - x_2}  \\
& \text{subject to} \quad &&\pi_1 + \pi_2 = 1,\\
& && \pi_1 x_1 + \pi_2 x_2 = c, \\
 & && \pi_1 (x_1-c)^2 + \pi_2 (x_2-c)^2 = \sigma^2,\\
 & &&  (\pi_1, \pi_2) \geq 0, \,\, (x_1, x_2) \leq  \alpha. \nonumber
\end{alignat*}
Assume without loss of generality that $x_2 > x_1$. Make the change
of variable $x_1 = c - \sigma t$, $x_2 = c + \sigma/t$,
$\pi_1 = \frac{1}{t^2 + 1}$, and $\pi_2 = \frac{t^2}{t^2 + 1}$, where
$t \geq 0$. The constraint becomes
$0 < t \leq \frac{\sigma}{\alpha - c}$ and the objective can be written
as a function of $t$, 
\[V(t) = \frac{1}{t^2+1} \frac{c - \sigma t }{ 1 - c +\sigma t } +
  \frac{t^2}{t^2+1} \frac{c + \sigma t }{ 1 - c - \sigma t }. \] The
derivative
\[V'(t) = -\frac{\sigma^{3}\left(t^{2}+1\right)}{((1 - c)
    t-\sigma)^{2}(1-c+\sigma t)^{2}}, \] is an increasing function of
$t$. Therefore, $V$ is convex and attains its maximum at
$t = \frac{\sigma}{\alpha - c}$, i.e., $x_2 = \alpha$.
Thus the optimal solution either has $x_1^{\star} = x_3^{\star}$ and $x_2^{\star} = \alpha$, or $x_2^{\star} = x_3^{\star} = \alpha$. Hence the set $\cb{x_1^{\star}, x_2^{\star}, x_3^{\star}}$ has at most two distinct elements, with one of them being $\alpha$. 

\paragraph{Combining both cases.} 
We showed that the optimal solution of \eqref{eqn:opt_3_terms} satisfies that the set $\cb{x_1^{\star}, x_2^{\star}, x_3^{\star}}$ has at most two distinct elements, with one of them being $\alpha$. Without loss of generality, assume $\pi_3^{\star} = 0$, $x_2^{\star} = \alpha$. It remains to solve
\begin{alignat*}{1}
& x_2 = \alpha,\\
& \pi_1 + \pi_2 = 1,\\
& \pi_1 x_1 + \pi_2 x_2 = c, \\
 &  \pi_1 (x_1-c)^2 + \pi_2 (x_2-c)^2 = \sigma^2. 
\end{alignat*}
The solution is given by $x_1 = c - \frac{\sigma^2}{\alpha - c}$,
$x_2 = \alpha$,
$\pi_1 = \frac{(\alpha - c)^2}{(\alpha - c)^2 + \sigma^2}$ and
$\pi_2 = \frac{\sigma^2}{(\alpha - c)^2 + \sigma^2}$. The optimal
value of \eqref{eqn:opt_3_terms} is thus 
\begin{align*} \Vopt & = \frac{(\alpha - c)^2}{(\alpha - c)^2 + \sigma^2} \frac{c - \frac{\sigma^2}{\alpha - c}}{ 1 - c + \frac{\sigma^2}{\alpha - c}} + \frac{\sigma^2}{(\alpha - c)^2 + \sigma^2} \frac{\alpha}{1 - \alpha} \\ & = \frac{\alpha}{1 - \alpha} - \frac{(\alpha - c)^2}{(1-\alpha)(\sigma^2 + (1-c)(\alpha - c))}.
\end{align*}
\end{proof}

\subsubsection{Asymptotic Sharpness of \eqref{eqn:exch_bound2}}
We give an example where the difference between the FDR and bound
\eqref{eqn:exch_bound2} converges to 0 as $m \to \infty$. Observe that 
\[\p{q + \frac{\delta}{1 + \beta \delta} \sqb{\frac{c}{1-c} - \frac{c - c\delta}{1 - (c-c\delta)}q}} \leq  \p{q + c(1-q)}\]
if and only $\rho \leq \frac{c(1-q)}{q + c(1-q)}$. We thus assume
$\rho \leq \frac{c(1-q)}{q + c(1-q)}$, and give an example where the FDR
matches 
$\p{q + \frac{\delta}{1 + \beta \delta} \sqb{\frac{c}{1-c} - \frac{c -
      c\delta}{1 - (c-c\delta)}q}}$ asymptotically.

Consider the case where there are $n_1 = \lfloor \sqrt{m} \rfloor$
nonnulls appearing first in the sequence. Let all the nonnull $\pvals$
be 0. We adopt the same notations as above section and set
$\alpha = \frac{c}{c + q - cq}$,
$x_1 = c - \frac{\sigma^2}{\alpha - c}$, $x_2 = \alpha$,
$\pi_1 = \frac{(\alpha - c)^2}{(\alpha - c)^2 + \tilde{\rho}(1-c)c}$
and $\pi_2 = \frac{\sigma^2}{(\alpha - c)^2 +
  \tilde{\rho}(1-c)c}$. Note that by definition of $\sigma$, we have that $\pi_1 + \pi_2 = 1$. The condition
$\rho \leq \frac{c(1-q)}{q + c(1-q)}$ ensures that
$x_1, x_2 \in [0,1]$. Let 
$n_0 = m - n_1$ be the number of null $\pvals$ and set $m_1$ =
$\lfloor n_0 x_1 \rfloor + 1$ and $m_2$ =
$\lfloor n_0 x_2 \rfloor - 2$. We consider null $\pvals$ with the following distribution: 
\begin{enumerate}
\item With probability $\pi_1$, pick $m_1$ indices uniformly at random
  from $\cb{1, \dots, m}$, and sample the corresponding $\pvals$ as
  $\text{i.i.d.}\Unif[0,c]$; sample the other $\pvals$ independently
  from $\Unif[c,1]$.
\item With probability $1 - \pi_1$, pick $m_2$ indices uniformly
  at random from $\cb{1, \dots, m}$, and sample the corresponding $\pvals$ as $\text{i.i.d.}\Unif[0,c]$; sample the other $\pvals$ independently from
  $\Unif[c,1]$.
\end{enumerate}
The null $\pvals$ sampled this way are clearly exchangeable. We then
show that the null $\pvals$ are stochastically larger than
$ \operatorname{Unif}[0,1]$ and that they satisfy
$\Corr{\one\cb{p_i \leq c}, \one\cb{p_j \leq c}} \leq \rho$. For a
null $j$,
$\PP{p_j \leq c} = m_1/(m-n_1) \pi_1 + m_2/(m-n_1) \pi_2 \leq x_1\pi_1
+ x_2 \pi_2 = c$. Hence, by construction,
$p_j \geq \operatorname{Unif}[0,1]$. Regarding the covariance,
note that
$$(n_0 - 1)n_0\Cov{\one\cb{p_i \leq c}, \one\cb{p_j \leq c}} = (m_2 -
m_1)^2\pi_1\pi_2 - n_0 \Var{\cb{p_i \leq c}}.$$ Thus
$$\Corr{\one\cb{p_i \leq c}, \one\cb{p_j \leq c}} = (m_2 -
m_1)^2\pi_1\pi_2/\p{n_0(n_0 -1)\Var{\cb{p_i \leq c}}} - 1/(n_0 -
1).$$ The variance term obeys 
$$\Var{\cb{p_i \leq c}} = (1 - (m_1\pi_1 + m_2\pi_2)/n_0)((m_1\pi_1 +
m_2\pi_2)/n_0).$$ Thus one can verify that
$$\Corr{\one\cb{p_i \leq c}, \one\cb{p_j \leq c}} \leq (x_2 -
x_1)^2\pi_1 \pi_2/(c(1-c)) - 1/(n_0-1) \leq \rho.$$

Since there are many nonnulls appearing early in the sequence with
vanishing $\pvals$, the set
$\cb{j \in \Ho, j \leq \hat{k}: p_j \leq c}$ is not empty. Thus the
super-martingale from \ref{subsection:proof_exch1} 
becomes a martingale after $\hat{k}$, whence, 
\[
\EE{ M(\hat{k}) \big| \mathcal{F}_m}
=\frac{\#\left\{ j \in \Ho: p_{j} \leq c\right\}}{1+\#\left\{j \in \Ho: p_{j}>c\right\}}. 
\]
Note also that since $\hat{k} \geq n_1$, we have 
\[\frac{1+\#\left\{ j  \leq \hat{k}: p_{j}>c\right\}}{\#\left\{j \leq \hat{k}: p_{j} \leq c\right\} \vee 1} \geq \frac{1 - c}{c} q - \frac{1 + \frac{1-c}{c}q}{n_1}. \]
Combining these two facts gives 
\begin{align*}
\EE{\FDP | \mathcal{F}_m} &= \mathbb{E}\left[\frac{V}{R \vee 1} \Big| \mathcal{F}_m \right]
=\mathbb{E}\left[\frac{V}{R \vee 1} \cdot \one\cb{\hat{k}>0} \Big| \mathcal{F}_m \right]\\
&=\mathbb{E}\left[\frac{\#\left\{j \in \Ho, j  \leq \hat{k}: p_{j} \leq c\right\}}{1+\#\left\{j \in \Ho, j  \leq \hat{k}: p_{j}>c\right\}} \cdot\left(\frac{1+\#\left\{j \in \Ho, j  \leq \hat{k}: p_{j}>c\right\}}{\#\left\{j \leq \hat{k}: p_{j} \leq c\right\} \vee 1} \cdot \one\cb{\hat{k}>0}\right) \Bigg| \mathcal{F}_m \right] \\
&=\mathbb{E}\left[\frac{\#\left\{j \in \Ho, j  \leq \hat{k}: p_{j} \leq c\right\}}{1+\#\left\{j \in \Ho, j  \leq \hat{k}: p_{j}>c\right\}} \cdot\left(\frac{1+\#\left\{j  \leq \hat{k}: p_{j}>c\right\}}{\#\left\{j \leq \hat{k}: p_{j} \leq c\right\} \vee 1} \cdot \right) \Bigg| \mathcal{F}_m \right] \\
&\geq \p{ \mathbb{E}\left[\frac{\#\left\{j \in \Ho, j  \leq \hat{k}: p_{j} \leq c\right\}}{1+\#\left\{j \in \Ho, j  \leq \hat{k}: p_{j}>c\right\}} \Bigg| \mathcal{F}_m \right] \cdot \p{\frac{1 - c}{c} q - \frac{1 + \frac{1-c}{c}q}{n_1}}} \\
&= \frac{\#\left\{ j \in \Ho: p_{j} \leq c\right\}}{1+\#\left\{j \in \Ho: p_{j}>c\right\}}  \cdot \p{\frac{1 - c}{c} q - \frac{1 + \frac{1-c}{c}q}{n_1}}. 
\end{align*}
Taking expectation on both hand sides, we get
\begin{align*}
\FDR \geq \p{ \frac{\pi_1m_1}{1 + m - n_1 - m_1} + \frac{\pi_2 m_2}{1 + m - n_1 - m_2} }  \p{\frac{1 - c}{c} q - \frac{1 + \frac{1-c}{c}q}{n_1}}. 
\end{align*}
Taking limits on both sides yields 
\begin{align*}
\liminf_{m \to \infty} \FDR &\geq \lim_{m \to \infty}\p{ \frac{\pi_1m_1}{1 + m - n_1 - m_1} + \frac{\pi_2 m_2}{1 + m - n_1 - m_2} }  \p{\frac{1 - c}{c} q - \frac{1 + \frac{1-c}{c}q}{n_1}}\\
& =  \lim_{m \to \infty}\p{ \frac{\pi_1 x_1}{1 - x_1} + \frac{\pi_2 x_2}{1 -x_2} } \frac{1 - c}{c} q \\
& = \lim_{m \to \infty} q + \frac{\tilde{\delta}}{1 + \beta \tilde{\delta}} \sqb{\frac{c}{1-c} - \frac{c - c\tilde{\delta}}{1 - (c-c\tilde{\delta})}q}, 
\end{align*}
where the last equality follows from the results
in \ref{subsection:proof_exch2}. Finally, note that by construction,
\[ \lim_{m \to \infty} q + \frac{\tilde{\delta}}{1 + \beta \tilde{\delta}} \sqb{\frac{c}{1-c} - \frac{c - c\tilde{\delta}}{1 - (c-c\tilde{\delta})}q} = q + \frac{\delta}{1 + \beta \delta} \sqb{\frac{c}{1-c} - \frac{c - c\delta}{1 - (c-c\delta)}q}.\]
Therefore,
\[\liminf_{m \to \infty} \FDR = q + \frac{\delta}{1 + \beta \delta} \sqb{\frac{c}{1-c} - \frac{c - c\delta}{1 - (c-c\delta)}q}. \]


\subsection{Proof of Theorem \ref{theo:cond_prob}}
We start by proving Theorem \ref{theo:cond_prob}. 
The proof follows the argument in \citet{barber2020robust}. Define
\[R_{\delta} = \frac{ \#\cb{j\leq \hat{k}, j \in \mathcal{H}_0 : p_j \leq c, a_j \leq c+ \delta}}{1 +   \#\cb{j\leq \hat{k}, j \in \mathcal{H}_0: p_j > c} }
=  \sum_{j \in \mathcal{H}_0} \frac{ \one\cb{j \leq \hat{k}}  \one\cb{p_j \leq c}  \one\cb{a_j \leq c + \delta}}
{1 + \sum_{i \in \mathcal{H}_0} \one\cb{i \leq \hat{k}}\one\cb{p_i > c} }. 
\]
Define $\hat{k}_j$ to be $\hat{k}$ if $p_j$ were at most $c$: formally, 
\[\hat{k}_j =\max \left\{k \in \cb{1, \dots, n}: \frac{1+\#\left\{i \leq k, i \neq j: p_i>c\right\}}
{\one\cb{j \leq k} +  \#\left\{i \leq k, i \neq j: p_i \leq c\right\} } \leq \frac{1-c}{c} \cdot q\right\}. \]
Observe he relation 
\[ \frac{ \one\cb{j \leq \hat{k}} \one\cb{p_j \leq c} \one\cb{a_j \leq
      c + \delta}}{1 + \sum_{i \in \mathcal{H}_0} \one\cb{i \leq
      \hat{k}}\one\cb{p_i > c} } = \frac{ \one\cb{j \leq \hat{k}_j}
    \one\cb{p_j \leq c} \one\cb{a_j \leq c + \delta}}{1 + \sum_{i \in
      \mathcal{H}_0, i \neq j} \one\cb{i \leq \hat{k}_j}\one\cb{p_i >
      c} }, \] since $p_j \leq c$ implies $\hat{k} = \hat{k}_j$.

The quantity $$\frac{ \one\cb{j \leq \hat{k}_j} \one\cb{a_j \leq c + \delta}}{1 + \sum_{i \in \mathcal{H}_0, i \neq j} \one\cb{i \leq \hat{k}_j}\one\cb{p_i > c}  }$$ does not depend on $p_j$, and only depends on $p_{-j}$. Hence, the expectation of $R_{\delta}$ can be written as 
\begin{align*}
\EE{R_{\delta}} &= \sum_{j \in \mathcal{H}_0} \EE{ \frac{ \one\cb{j \leq \hat{k}}  \one\cb{p_j \leq c} \one\cb{a_j \leq c + \delta}}
{1 + \sum_{i \in \mathcal{H}_0} \one\cb{i \leq \hat{k}}\one\cb{p_i > c}  }  }\\
& = \sum_{j \in \mathcal{H}_0} \EE{ \frac{ \one\cb{j \leq \hat{k}_j}  \one\cb{p_j \leq c}\one\cb{a_j \leq c + \delta} }
{1 + \sum_{i \in \mathcal{H}_0, i \neq j} \one\cb{i \leq \hat{k}_j}\one\cb{p_i > c}  } }\\
& = \sum_{j \in \mathcal{H}_0} \EE{ a_j  \frac{ \one\cb{j \leq \hat{k}_j} \one\cb{a_j \leq c + \delta}  }
{1 + \sum_{i \in \mathcal{H}_0, i \neq j} \one\cb{i \leq \hat{k}_j}\one\cb{p_i > c}  } }\\
& \leq (c+\delta) \sum_{j \in \mathcal{H}_0} \EE{  \frac{ \one\cb{j \leq \hat{k}_j} \one\cb{a_j \leq c + \delta}  }
{1 + \sum_{i \in \mathcal{H}_0, i \neq j} \one\cb{i \leq \hat{k}_j}\one\cb{p_i > c}  } }.
\end{align*}

The term can be further bounded by
\begin{multline*}
 (c+\delta) \Bigg[ \sum_{j \in \mathcal{H}_0} \EE{  \frac{ \one\cb{j \leq \hat{k}_j}  \one\cb{p_j \leq c}  \one\cb{a_j \leq c + \delta}  }
{1 + \sum_{i \in \mathcal{H}_0, i \neq j} \one\cb{i \leq \hat{k}_j}  \one\cb{p_i> c}  } }
 + \sum_{j \in \mathcal{H}_0} \EE{  \frac{ \one\cb{j \leq \hat{k}_j}  \one\cb{p_j > c}   }
   {1 + \sum_{i \in \mathcal{H}_0, i \neq j} \one\cb{i \leq \hat{k}_j}  \one\cb{p_i > c}} } \Bigg]\\
= (c+\delta) \sqb{\EE{R_{\delta}} + \sum_{j \in \mathcal{H}_0} \EE{  \frac{ \one\cb{j \leq \hat{k}_j}  \one\cb{p_j > c}   }
{1 + \sum_{i \in \mathcal{H}_0, i \neq j} \one\cb{i \leq \hat{k}_j}  \one\cb{p_i > c}} } }. 
\end{multline*}
For the second term, unless the numerator is zero, 
\begin{align*}
 \sum_{j \in \mathcal{H}_0}  \frac{ \one\cb{j \leq \hat{k}_j}  \one\cb{p_j > c}   }
{1 + \sum_{i \in \mathcal{H}_0, i \neq j} \one\cb{i \leq \hat{k}_j}  \one\cb{p_i > c}} 
 &= \sum_{j \in \mathcal{H}_0}  \frac{ \one\cb{j \leq \hat{k}_j}  \one\cb{p_j > c}   }
{1 + \sum_{i \in \mathcal{H}_0, i \neq j} \one\cb{i \leq \hat{k}_i}  \one\cb{p_i > c}}\\
& =  \frac{ \sum_{j \in \mathcal{H}_0}  \one\cb{j \leq \hat{k}_j}  \one\cb{p_j > c}   }
{ \sum_{i \in \mathcal{H}_0} \one\cb{i \leq \hat{k}_i}  \one\cb{p_i > c}} = 1. \\
\end{align*}

Combining the results above, we have
$\EE{R_{\delta}} \leq (c + \delta) (\EE{R_{\delta}} + 1)$ and hence
\[\EE{R_{\delta}} \leq \frac{c+\delta}{1-c-\delta}. \]

Letting $\hat{\mathcal{S}}$ be the set of
rejections, we have 
\begin{align*}
\frac{\abs{j : j \in \hat{\mathcal{S}} \cap \mathcal{H}_0 \text{ and } a_j \leq c+\delta}} {\abs{\hat{\mathcal{S}}} \vee 1} 
& = \frac{ \#\cb{j\leq \hat{k}, j \in \mathcal{H}_0: p_j \leq c, a_j \leq c+ \delta}}{\#\left\{j \leq \hat{k}, j \in \mathcal{H}_0: p_{j} \leq c\right\} \vee 1}\\
& = 
\frac{ \#\cb{j\leq \hat{k}, j \in \mathcal{H}_0: p_j \leq c, a_j \leq c+ \delta}}{1 +   \#\cb{j\leq \hat{k}, j \in \mathcal{H}_0: p_j > c} }  \frac{1+\#\left\{j \leq \hat{k}, j \in \mathcal{H}_0: p_{j}>c\right\}}{\#\left\{j \leq \hat{k}: p_{j} \leq c\right\} \vee 1} \\
& \leq  R_{\delta} \frac{1-c}{c}  q. 
\end{align*}
Since the $\FDP$ is at most $1$, 
the FDR is at most 
\[
\FDR = \EE{\frac{\abs{j : j \in \hat{\mathcal{S}} \cap \mathcal{H}_0} }{ \abs{\hat{\mathcal{S}}} \vee 1}}
\leq \frac{1-c}{c}  q \EE{R_{\delta}} + \PP{ \max_{j \in \mathcal{H}_0} a_{j} > c + \delta}
\leq q \frac{c+\delta}{c}\frac{1-c}{1-c-\delta} + \epsilon. 
\]

\subsection{Proof of Theorem \ref{theo:no_assu}}
By definition, 
\[\FDR 
= \EE{\frac{\#\cb{\text{null } j \leq \hat{k}: p_{j} \leq c}}{\#\cb{j \leq \hat{k}: p_{j} \leq c}\vee 1}}
= \sum_{j \in \mathcal{H}_0} \EE{\frac{\one\cb{j \leq \hat{k}} \one\cb{p_j \leq c}}{\#\cb{j \leq \hat{k}: p_{j} \leq c}\vee 1}}.\]
For each null $j$, 
\begin{align*}
 \EE{\frac{\one\cb{j \leq \hat{k}} \one\cb{p_j \leq c}}{\#\cb{j \leq \hat{k}: p_{j} \leq c}\vee 1}}
 &\leq \PP{p_j \leq c} \max\cb{\frac{\one\cb{j \leq \hat{k}}}{\#\cb{j \leq \hat{k}: p_{j} \leq c}\vee 1}}  \\
  & \leq c \max\cb{\frac{\one\cb{j \leq \hat{k}}}{\#\cb{j \leq \hat{k}: p_{j} \leq c}\vee 1}} .
\end{align*}
We study the quantity  
\[\frac{\one\cb{j \leq \hat{k}}}{\#\cb{j \leq \hat{k}: p_{j} \leq c}\vee 1}. \]
By definition of $\hat{k}$, unless $\hat{k} = 0$, 
\[\frac{1+\#\cb{j \leq \hat{k}: p_{j}>c}}{\#\cb{j \leq \hat{k}: p_{j} \leq c}\vee 1} \leq \frac{1-c}{c} \, q.\]
Hence, as long as $\frac{1-c}{c} q < 1$, 
\[ \frac{1 + \hat{k} - \#\cb{j \leq \hat{k}: p_{j} \leq c}}{\#\cb{j \leq \hat{k}: p_{j} \leq c}} \leq \frac{1 - c}{c}q.  \]
This implies that 
\[\#\cb{j \leq \hat{k}: p_{j} \leq c} \geq \frac{1 + \hat{k}}{1 + \frac{1-c}{c}q}. \]
Therefore, 
\[\frac{\one\cb{j \leq \hat{k}}}{\#\cb{j \leq \hat{k}: p_{j} \leq c}\vee 1} \leq \frac{\one\cb{j \leq \hat{k}}}{1 + \hat{k}}\p{1 + \frac{1-c}{c}q} \leq \frac{1 + \frac{1-c}{c}q}{1 + j}. \]
Note that the above quantity holds for $\hat{k} = 0$ as well. This shows that 
\begin{align*}
\EE{\frac{\one\cb{j \leq \hat{k}} \one\cb{p_j \leq c}}{\#\cb{j \leq \hat{k}: p_{j} \leq c}\vee 1}}
 &\leq 
c \max\cb{\frac{\one\cb{j \leq \hat{k}}}{\#\cb{j \leq \hat{k}: p_{j} \leq c}\vee 1}}\\
& \leq \p{c + (1-c) q} \frac{1}{1 + j}. 
\end{align*}
Taking the summation over null $j$'s gives the desired result. 

\end{appendix}

\end{document}